\documentclass[twoside]{article}

\usepackage[accepted]{aistats2022}
\usepackage[round]{natbib}

\usepackage{enumerate}
\usepackage{graphicx}
\usepackage{amssymb}
\usepackage{mathtools}
\usepackage{xspace}
\usepackage{subfigure}
\usepackage[lined,linesnumbered,ruled,commentsnumbered]{algorithm2e} % for algorithm
\usepackage{colonequals}
\usepackage[mathscr]{euscript}
\usepackage{fixltx2e}    % for text subscript
\usepackage{amsthm}
\usepackage{mathrsfs}
\usepackage{hyperref}
\usepackage{caption}
\usepackage{float}
\usepackage{relsize}
\usepackage{dsfont}
\usepackage{tikz}
\usepackage{pgfplots}
\usetikzlibrary{arrows,shapes,backgrounds,plotmarks,positioning}
\pgfplotsset{compat=newest}                         % move axis labels close to the tick label automatically
\pgfplotsset{plot coordinates/math parser=false}
\newtheorem{theorem}{Theorem}%[section]
\newtheorem{lemma}{Lemma}
\newtheorem{definition}{Definition}

\newtheorem{proposition}{Proposition}

\newtheorem{remark}{Remark}

\newcommand*\dif{\mathop{}\!\mathrm{d}}

\newcommand{\SA}{\mathcal{S}}
\newcommand{\X}{\mathcal{X}}

\newcommand{\N}{\mathcal{V}}

\newcommand{\VAR}{\text{VAR}}

\newcommand{\Real}{\mathds{R}}

\newcommand{\SetPair}{\mathds{S}}
\newcommand{\D}{\mathds{D}}
\newcommand{\E}{\mathds{E}}

\newcommand{\Set}[1]{\{#1\}}

\newcommand{\supp}{\text{supp}}

\begin{document}



\twocolumn[

\aistatstitle{Kantorovich Mechanism for Pufferfish Privacy}

\aistatsauthor{Ni Ding}

\aistatsaddress{University of Melbourne} ]

\begin{abstract}
    Pufferfish privacy achieves $\epsilon$-indistinguishability over a set of secret pairs in the disclosed data. This paper studies how to attain $\epsilon$-pufferfish privacy by exponential mechanism, an additive noise scheme that generalizes the Laplace noise. It is shown that the disclosed data is $\epsilon$-pufferfish private if the noise is calibrated to the sensitivity of the Kantorovich optimal transport plan. Such a plan can be obtained directly from the data statistics conditioned on the secret, the prior knowledge of the system. The sufficient condition is further relaxed to reduce the noise power. It is also proved that the Gaussian mechanism based on the Kantorovich approach attains the $\delta$-approximation of $\epsilon$-pufferfish privacy.
\end{abstract}

\section{INTRODUCTION}

Data privacy is about how to protect the confidential information when people are sharing data with each other.
For a data user that is able to analyzes the disclosed dataset and computes statistics about it, the purpose is to prevent any prediction on the sensitive attributes or secrets, e.g., ID, age, gender, race, etc., that could be maliciously used for discrimination or unfair decision making.
Thus, data protection nowadays is far beyond anonymization, but more of \emph{inference control} \citep{Dwork2008Survey}.

Differential privacy \citep{Dwork2014book} ensures the statistical indifference about the secret when the adversary is collecting the aggregated statistics about the underlying population. For all neighboring databases $s_i$ and $s_j$ that differ only in one entry, $\epsilon$-differential privacy upper bounds the statistical distance between the probabilities $\Pr(\tilde{f}(s_i)=y)$ and $\Pr(\tilde{f}(s_j)=y)$ by a positive threshold $\epsilon$. Here, $\tilde{f}$ denotes the noised, or randomized, query function $f$.
Differential privacy is a rigorous mathematical definition of privacy that can be easily applied to machine learning tasks, e.g., the contingency table release \citep{Barak2007Contingency}, the privacy-preserving data mining~\citep{Dwork2004PivDM}.

On the other hand, information theorists study data privacy in a Bayesian inference setting~\citep{PvsInfer2012,Lalitha2013TIFS,Ding2020TIFS}. The secret is treated as a random variable $S$ that is correlated with the public data $X$ to be released.
For $\Pr(S=s)$ being the adversary's prior belief of the secret, the purpose is to randomize $X$ and generate the sanitized data $Y$ to reduce the information gain on $S$, i.e., $Y$ should reduce the difference between the posterior belief $\Pr(S=s|Y=y)$ and prior belief $\Pr(S=s)$.
This can be translated, by Bayes' rule, to bounding the statistical distance of the conditional probabilities $\Pr(Y=y|S=s_i)$ and $\Pr(Y=y|S=s_j)$ by $\epsilon$ for each pair of secret instances $s_i$ and $s_j$, which is called $\epsilon$-local differential privacy~\citep{RDP2017,MaxL2020,Ding2021ISIT}.

\emph{Pufferfish privacy}:
\citet{Pufferfish2012KiferConf,Pufferfish2014Kifer} introduced a more general privacy framework called `pufferfish'.
For $\rho$ being the prior knowledge of the system, $\epsilon$-pufferfish privacy enforces statistical indistinguishability between $\Pr(Y=y|S=s_i,\rho)$ and $\Pr(Y=y|S=s_j,\rho)$ over all pairs of secrets $(s_i,s_j)$ in a discriminative pair set $\SetPair$.
This is a more flexible and practical setting in that: $\SetPair$ can be specified by all secret pairs that actually raise the privacy concerns in the real application; $\rho$ could denote the side information obtained by the adversary, which incorporates the concept of Bayesian inference in information-theoretic data privacy.
It is also shown by \citet{Pufferfish2014Kifer} that, for specific $\SetPair$ and $\rho$, $\epsilon$-pufferfish privacy is equivalent to $\epsilon$-differential privacy \citep{Dwork2006} and $\epsilon$-indistinguishability \citep{CalibNoiseDP}.

\emph{Privatization}:
The question then is how to privatize the data to attain pufferfish privacy.
The information-theoretic solution is to determine a privacy-preserving encoding function $\Pr(Y=y|X=x)$ for each piece of message $x$ and codeword $y$ \citep{Makhdoumi2014PF}. But, this is only practical for discrete and finite alphabet and not as convenient as the additive noise (e.g., Laplace) mechanism, where we only need to calibrate the parameter of the noise distribution.
There are other attempts in the literature, e.g., segmented noise mechanism for publishing counting query and histogram~\citep{Pufferfish2014Kifer}, Laplace mechanism for monitoring web browsing behavior~\citep{Liang2020Web}.
However, these methods only apply to specific applications, e.g., a particular query function, Markovian assumption about the prior knowledge $\rho$.

The \emph{challenge} here is that the statistics of $Y$ is caused by the randomness in the data regulation scheme, as well as the intrinsic correlation between $S$ and $X$, i.e., the conditional probability $\Pr(X=x|S=s,\rho)$.\footnote{This is the reason that pufferfish privacy is considered a generalization of differential privacy for correlated data. An equivalent problem is how to attain differential privacy when the query answer $f$ itself is a randomized function, for which the noise calibration by $\ell_1$-sensitivity \citep{CalibNoiseDP} does not directly apply. See Section~\ref{sec:main}. }
\cite{PufferfishWasserstein2017Song} proposed an additive noise scheme based on the Wasserstein metric in the probability space: calibrating Laplace noise to the maximum $\infty$-Wasserstein distance over all secret pairs in $\SetPair$ attains pufferfish privacy.
While this method generally applies to any system, computing the $\infty$-Wasserstein distance is hard.\footnote{This is due to the difficulty (non-convexity) in obtaining the optimal transport plan for the $\infty$-Wasserstein metric. See~\cite{Champion2008InfW,DePascale2019InfW}.}
\cite{PufferfishWasserstein2017Song} then resorted to a Markov quilt mechanism, which does not require Wasserstein metric, but only works in Bayesian network models.

\subsection{Our Contributions}

In this paper, we propose a Kantorovich mechanism for attaining pufferfish privacy that generally applies to any data $S$ and $X$ and the prior knowledge $\rho$.

Our main contributions are the following.

\begin{enumerate}
    \item We consider the exponential mechanism, an additive noise scheme that generalizes Laplace noise.
     A sufficient condition is derived showing that pufferfish privacy is attained by calibrating noise to the sensitivity of the Kantorovich optimal transport plan. This transport plan can be directly determined by the conditional probabilities $\Pr(X=x|S=s_i,\rho)$ and $\Pr(X=x|S=s_i,\rho)$  for all pairs of secrets $s_i$ and $s_j$ in $\SetPair$.
     It is also proved that the Gaussian mechanism based on this Kantorovich approach attains $(\epsilon,\delta)$-pufferfish privacy.

    \item We relax the sufficient condition to reduce the noise power. Experimental results show that the relaxed sufficient condition improves data utility of the pufferfish private data regulation schemes.

    \item In a multi-user system, where each user is assigned an independent random variable, we study the $\epsilon$-indistinguishability as to whether a user is present in the system and the value of the random variable he/she obtains. It is shown that, for any deterministic query function $f$, the sensitivity of the Kantorovich optimal transport plan is equivalent to that of $f$, regardless of the randomness in the prior knowledge $\rho$.
        Therefore, pufferfish privacy is attained by calibrating noise to the sensitivity of $f$.
\end{enumerate}

In this paper, we only present a proof sketch for each statement (incl. theorem, lemma and corollary). The full proof and detailed derivation can be found in the supplementary materials.
We use capital letters, e.g, $X$, to denote a random variable and lower case letters, e.g., $x$, to denote the instance of this random variable. Notation $P_{X}(x)$ denotes the probability $\Pr(X = x)$.

\section{PUFFERFISH PRIVACY}

Let $\SA$ be the alphabet of the secret $S$ and $Y$ be the sanitized version of $X$.
We say $Y$ is private if it attains a certain degree of statistical indistinguishability of the sensitive information $S$.
Here, the indistinguishability refers to the bounded probability of $Y$ evoked by two secrets $s_i, s_j \in \SA$:
\begin{equation} \label{eq:LDP}
    \left| \log \frac{P_{Y|S}(y|s_i)}{P_{Y|S}(y|s_j)} \right| \leq \epsilon
\end{equation}
for some $\epsilon > 0 $. We call $\epsilon$ the \emph{privacy budget}.
If \eqref{eq:LDP} holds for all secret pairs $(s_i,s_j) \in \SA^2$, $Y$ is $\epsilon$-local differentially private~\citep{LDP2013MiniMax,LDP2014Lalitha}; if it holds for all neighboring $(s_i,s_j)$,\footnote{The neighborhood is defined by Hamming distance: $d_{H}(s_i,s_j) \leq 1$}
$Y$ attains $\epsilon$-differential privacy \cite{CalibNoiseDP}.

%\cite{Pufferfish2014Kifer} introduces a more general privacy framework called `pufferfish'.
Let $\SetPair \subseteq \SA^2$ be the \emph{discriminative pair set} containing secret pairs $(s_i,s_j)$.
Pufferfish privacy attains $\epsilon$-indistinguishability in $\SetPair$.
%
%We call $\SetPair$ the \emph{discriminative pair set}, which can be specified by the customer or data curator based on the real application.
%
%Pufferfish privacy is defined as follows.

\begin{definition}[Pufferfish Privacy {\citep{Pufferfish2014Kifer}}]
The sanitized data $Y$ attains $(\epsilon,\SetPair)$-pufferfish privacy if
\begin{equation} \label{eq:DefPufferfish}
	 \Big| \log \frac{P_{Y|S}(y|s_i, \rho )}{P_{Y|S}(y|s_j, \rho)} \Big| \leq \epsilon,
\end{equation}
for all $\rho \in \D$ and $(s_i,s_j) \in \SetPair$.
\end{definition}

Here, $\rho$ denotes the prior knowledge, e.g., the hyperparameter, that is sufficient to describe the probability distribution of $X$ given the secret $S$, which we denote by $P_{X|S}(x|s,\rho)$.
$\rho$ could be the true knowledge of $P_{X|S}(x|s,\rho)$. Or, if there are more than one adversary in the system, each $\rho$ can be used to denote the prior belief of an adversary and $\D$ is the set containing all adversaries in the system.
In this case, the $(\epsilon,\SetPair)$-pufferfish privacy guarantees the statistical indistinguishability \eqref{eq:DefPufferfish} against all adversaries $\rho \in \D$.

We show two examples of $S$, $X$ and $\rho$ below.
They will be used to present the main results in this paper.

\paragraph{$V$ independent user system}
Denote $\N$ the finite set that indexes $V = |\N|$ users or participants. Let each user throw a dice $S_i$ independently and denote the outcome by a multiple random variable $S = (S_i \colon i \in \N)$.
Assume $X = f(S)$, where $f$ is a deterministic query function.
If $S_i \in \Set{0,1}$ for all $i \in \N$, $X = f(S) = \sum_{i\in \N} S_i$ is a vote counting function. In this case, let $\rho = (p_i \colon i \in \N)$ with each $p_i$ denoting the probability of the event $S_i = 1$ for a Bernoulli distribution.\footnote{In this case, each $p_i$ could denote the local randomization scheme chosen by individual $i$, as in local differential privacy \cite{LDP2013MiniMax}. }
The conditional probability $P_{X|S}(\cdot|s_i,\rho)$ is fully determined by $\rho$.

\paragraph{Attributes in tabular data}

A tabular dataset needs to be denonymized or privatized before being disclosed to protect the sensitive attributes/columns, e.g., `name', `age', `race'. Denote $S$ and $X$ the sensitive and public attributes,\footnote{Or, $X$ could denote some deterministic function of the public attribute, as in \cite{Pufferfish2014Kifer,PufferfishWasserstein2017Song}. }
respectively.
This is the information-theoretic data privacy problem formulated in \citet[Fig.~1]{Lalitha2013TIFS}.
Let $\rho$ be the empirical joint distribution of $S$ and $X$, which determines the conditional probability $P_{X|S}(\cdot|s_i,\rho)$.
In this case, we can still write $X = f(S)$, while $f$ is not a deterministic but a randomized function.

\subsection{Additive Noise Privatization Scheme}

Consider the additive noise mechanism
$$ Y = X + N$$
where the noise $N$ is independent of $X$.
Denote the probability of $N$ by $P_N(\cdot)$. The conditional probability $P_{Y|S}(\cdot|s,\rho)$ is determined by the convolution
\begin{equation} \label{eq:Convolve}
    P_{Y|S}(y|s,\rho) = \int P_N(y-x) P_{X|S}(x|s,\rho) \dif x.
\end{equation}

For zero-mean noise, we have $\E[Y] = \E[X]$ and the noise variance determines mean square error (MSE) $\E [ (Y - X)^2 ] = \VAR[N]$.
Here, the MSE denotes the average distortion of the released data $Y$, which can be used to quantify the data utility loss, e.g., \citet{Blowfish2014}.
It is clear that for two additive noise mechanisms both attaining pufferfish privacy, the one with less noise power is superior to the other.
In this paper, we use the notation $N_{\theta}$ for noise, where $\theta$ denotes the parameter that determines the probability density function $P_{N_{\theta}}(\cdot)$.

\section{KANTOROVICH MECHANISM}

In this section, we propose a Kantorovich approach, based on the $1$-Wasserstein metric, for attaining the pufferfish privacy.
We first convert the $\infty$-Wasserstein metric in the existing randomization mechanism \citep{PufferfishWasserstein2017Song} to a $1$-Wasserstein distance and propose the Kantorovich solution for Laplace noise.
Then, we extend this result to the exponential mechanism that generalizes the Laplace noise.

%a sufficient condition and its relaxation are derived, based on which, we propose the Kantorovich for attaining pufferfish privacy.

\subsection{Preliminaries}

We introduce the notation and definition for the Wasserstein metric and review the Kantororivich optimal transport plan as follows.
For each $\rho$, a joint distribution $\pi \colon \Real^2 \mapsto [0,1]$ is called a \emph{coupling} of $P_{X|S}(\cdot|s_i,\rho)$ and $P_{X|S}(\cdot|s_j,\rho)$ if they are the marginals of $\pi$.\footnote{That is, $\int \pi(x,x')\dif x' = P_{X|S}(s|s_i), \forall x$ and $\int \pi(x,x') \dif x  = P_{X|S}(x'|s_j), \forall x'$.}
Denote $\Gamma(s_i, s_j)$ the set of all couplings for the secret pair $(s_i,s_j) \in \SetPair$.
The $\alpha$th Wasserstein distance is $ W_{\alpha}(s_i, s_j) = \big( \inf_{\pi \in \Gamma(s_i,s_j)} \int d^{\alpha}(x-x') \dif \pi(x,x') \big)^{1/\alpha}$.
For $\alpha = 1$,
\begin{equation} \label{eq:Kantorovich}
    W_{1}(s_i, s_j) = \inf_{\pi \in \Gamma(s_i,s_j)} \int d(x-x') \dif \pi(x,x')
\end{equation}
is called the Kantorovich transportation problem for the mass transport cost $d$.

\paragraph{Kantorovich optimal transport plan $\pi^*$}\citep{Villani2009OPT,Santambrogio2015OPT}
For convex $d$, the minimizer $\pi^*$ of \eqref{eq:Kantorovich} can be directly determined by $P_{X|S}(\cdot|s_i,\rho)$ and $P_{X|S}(\cdot|s_j,\rho)$:
\[ \pi^*(x,x') = \frac{\dif^2}{\dif x  \dif x'} \min \big\{ F_{X|S}(x|s_i,\rho), F_{X|S}(x'|s_j,\rho) \big\}, \]
where $F_{X|S}(\cdot|s_i,\rho)$ and $F_{X|S}(\cdot|s_j,\rho)$ are the cumulative density functions for $P_{X|S}(\cdot|s_i,\rho)$ and $P_{X|S}(\cdot|s_j,\rho)$, respectively.

\subsection{Kantorovich-Laplace mechanism}

For Laplace noise $N_{\theta}$ with the noise distribution $P_{N_{\theta}}(z) = \frac{1}{2\theta} e^{-\frac{|z|}{\theta}}$, it is shown in \citet{PufferfishWasserstein2017Song} that calibrating noise power to
\begin{equation} \label{eq:WInftyLaplace}
    \theta = \frac{1}{\epsilon}\max_{\rho \in \D, (s_i,s_j) \in \SetPair}  W_\infty(s_i, s_j)
\end{equation}
for the $\ell_1$-norm $d(z) = |z|$ attains $(\epsilon,\SetPair)$-pufferfish privacy in $Y$.
Because the minimizer of $W_\infty(s_i, s_j) = \inf_{\pi \in \Gamma(s_i,s_j)} \sup_{(x,x') \in \supp(\pi)} |x-x'|$ is hard to obtain~\citep{Champion2008InfW,DePascale2019InfW}, we convert it to a $W_1$ metric and propose a Kantorovich approach below.

\begin{lemma}[From $W_\infty$ to Kantorovich] \label{lemma:Winfty2W1Laplace}
     Adding Laplace noise $N_{\theta}$ with
     \begin{equation} \label{eq:W1Laplace}
         \theta = \frac{1}{\epsilon} \max_{\rho \in \D, (s_i,s_j) \in \SetPair} \sup_{(x,x')\in \supp(\pi^*)} |x-x'|
     \end{equation}
     attains $(\epsilon, \SetPair)$-pufferfish privacy in $Y$.
\end{lemma}
\begin{proof}
We have $\theta$ in \eqref{eq:WInftyLaplace} equal to
\begin{equation} \label{eq:lemma:Winfty2W1Laplace}
\max_{\substack{\rho \in \D,\\ (s_i,s_j) \in \SetPair}} \Big\{ \theta \colon \inf\limits_{\pi \in \Gamma(s_i,s_j)} \int \Big[ \frac{|x-x'|}{\theta} - \epsilon \Big]_+ \dif \pi(x,x') = 0 \Big\}
\end{equation}
where $[z]_+ = \max\Set{z,0}$ is convex.
For each $\theta$, the minimizer of $\inf\limits_{\pi \in \Gamma(s_i,s_j)} \int \Big[ \frac{|x-x'|}{\theta} - \epsilon \Big]_+ \dif \pi(x,x')$ is the Kantorovich optimal transport plan $\pi^*$. Therefore, the maximum value of $\theta$ in \eqref{eq:lemma:Winfty2W1Laplace} equals to \eqref{eq:W1Laplace}.
\end{proof}

It is clear in the proof of Lemma~\ref{lemma:Winfty2W1Laplace} (see the full proof in Section~\ref{app:lemma:Winfty2W1Laplace} in the supplementary material) that the Wasserstein mechanism in the order of $\alpha = \infty$ proposed in \citet{PufferfishWasserstein2017Song} is equivalent to the Kantorovich-Laplace mechanism in Lemma~\ref{lemma:Winfty2W1Laplace}.
See Section~\ref{app:lemma:Winfty2W1LaplaceEx} in the supplementary material how to efficiently compute the solutions to the two examples in~\cite{PufferfishWasserstein2017Song} by Lemma~\ref{lemma:Winfty2W1Laplace}.

\subsection{Exponential Mechanism}
Let $d$ be a metric, i.e., $d$ is nonnegative, symmetric $d(z) = d(-z), \forall z$, and satisfies the triangular inequality $ d(z) \leq d(\delta) + d(z-\delta), \forall z, \delta$.
Consider the exponential mechanism \citep[Section~3.3]{CalibNoiseDP}, where the noise distribution is characterized by an exponential function $P_{N_\theta} (z) \propto e^{ -\eta(\theta) d(z)}$.
Assume $\eta \propto \frac{1}{\theta}$.
For Laplace mechanism, $\eta(\theta) = 1/\theta$ and $d(z) = |z|$.
By the triangular inequality,
\begin{equation}
	P_{N_{\theta}}(y-x) \leq e^{\eta(\theta) d(x-x')}  P_{N_{\theta}}(y-x') , \quad \forall x,x',y
\end{equation}
where $e^{\eta(\theta) d(x-x')}$ denotes an upper bound on the probability mass transport cost from $x$ to $x'$.
This cost upper bound is used in \citet{CalibNoiseDP,Dwork2014book} to prove that the differential privacy is attained by calibrating the standard deviation of the noise to the sensitivity of the query function $f$.
We obtain a similar result for attaining the pufferfish privacy in the $V$ independent user system as follows.

\subsubsection{Calibrating Noise to Sensitivity in $V$ independent User System }

For each user $i$, denote $S_i = a$ the event that user $i$ is present in the system and reports the value $a$ of random variable (i.e., the dice face) of $S_i$.
We write $S_i = \perp_i$ for the event when user $i$ is absent in the system.
Consider the following two discriminative pair sets
\begin{align*}
    &\SetPair_i = \Set{(S_i = a, S_i = b) \colon a, b \in \SA_i},  \\
    &\SetPair_{\perp_i} = \Set{(S_i = a, S_i = \perp_i) \colon a \in \SA_i}.  \\
\end{align*}
where $\SA_i$ denotes the alphabet of $S_i$ for user $i$.
Using $\SetPair_i$ and $\SetPair_{\perp_i}$, \citet{Pufferfish2014Kifer} proved that $(\epsilon,\SetPair)$-pufferfish privacy is equivalent to $\epsilon$-differential privacy \citep{Dwork2006} and $\epsilon$-indistinguishability \citep{CalibNoiseDP}, respectively.

Denote $S_{-i} = (S_{i'} \colon i' \in \N \setminus \Set{i})$ the multiple random variable excluding dimension/user $i$.
For the deterministic query function $f$, let the sensitivity for metric $d$ over the discriminative pair set $\SetPair_i$ be
$$ \triangle_f(\SetPair_i) = \max_{a,b\in \SA_i} \max_{s_{-i}} d \big( f(s_i=a ,s_{-i}) - f(s_i=b,s_{-i}) \big). $$
Note that for deterministic function $f$, the sensitivity $\triangle_f(\SetPair_i)$ is independent of $\rho$.

\begin{lemma} \label{lemma:NindUser}
    In the $V$ independent user system, for any deterministic query function $f$ and any prior knowledge $\rho$, the following holds for all $i \in \N$:
    \begin{enumerate}[(a)]
      \item adding noise $N_{\theta}$ with $\theta = \eta^{-1}(\epsilon / \triangle_f(\SetPair_i))$ attains $(\epsilon,\SetPair_i)$-pufferfish privacy in $Y$,
      \item if $Y$ is $(\epsilon,\SetPair_i)$-pufferfish private, then it is also $(\epsilon,\SetPair_{\perp_i})$-pufferfish private.
    \end{enumerate}
\end{lemma}
\begin{proof}
For any $a,b \in \SA_i$, we have the statistical indistinguishability in $Y$ bounded by the mass transport cost upper bound function:
$P_{Y|S_i}(y|a,\rho) \leq e^{\eta(\theta) d(f(s_i=a,s_{-i})-f(s_i=b,s_{-i})) } P_{Y|S_i}(y|b,\rho)$.
Then, for all $a,b \in \SA_i$, $y$ and $\rho$,
$P_{Y|S_i}(y|a,\rho) \leq e^{\eta(\theta) \triangle_f(\SetPair_i)} P_{Y|S_i}(y|b,\rho)$. Therefore, (a) is a sufficient condition for attaining $(\epsilon, \SetPair_i)$-pufferfish privacy.

Using the fact that $ \min_{s_i} P_{Y|S_i}(y |s_i,\rho) \leq P_{Y|S_i}(y | \perp_i,\rho) \leq \max_{s_i} P_{Y|S_i}(y |s_i,\rho)$,
for $(\epsilon,\SetPair_i)$-pufferfish private $Y$, we have $|\log \frac{P_{Y|S_i}(y|\perp_i,\rho)}{P_{Y|S_i}(y|a,\rho)}| \leq \max_{s_i} |\log  |\frac{P_{Y|S_i}(y|s_i,\rho)}{P_{Y|S_i}(y|a,\rho)}| \leq \epsilon$, i.e., $(\epsilon,\SetPair_{\perp_i})$-pufferfish privacy attains simultaneously.
\end{proof}

See Sections~\ref{appA:lemma:NindUser} and \ref{appB:lemma:NindUser} in the supplementary material for the detailed derivation and proof.
In Section~\ref{sec:Lemma2Verify}, we will verify Lemma~\ref{lemma:NindUser} by the Kantorovich optimal transport plan $\pi^*$, where it is revealed that $\triangle_f(\SetPair_i)$ coincides with the sensitivity of $\pi^*$.

\subsection{Sufficient condition}
\label{sec:main}

Following Lemma~\ref{lemma:NindUser}, for $f$ being a randomized function, adding noise $N_{\theta}$ with $\theta = \eta^{-1}(\epsilon/\triangle_f(\SetPair))$ also attains pufferfish privacy.
However, we should take into account the domain of $f$ as well as the randomness in $\rho$.
That is, for $X = f(S)$ where $f$ is a randomized function, the sensitivity of $f$ for metric $d$ over the discriminative pair set $\SetPair$ is
\begin{equation} \label{eq:SensitivityRandom}
    \triangle_f(\SetPair) =
    \max_{(s_i,s_j)\in \SetPair} \max_{\substack{x \in \supp(P_{X|S}(\cdot|s_i,\rho)),\\x' \in \supp(P_{X|S}(\cdot|s_j,\rho))}} d \big( x - x' \big).
\end{equation}
The probability mass $P_{X|S}(\cdot|s,\rho)$ could spread over a wide range of $X$ that significantly increases $\triangle_f(\SetPair)$.
Section~\ref{sec:exp} shows an example when $\supp(P_{X|S}(\cdot|s_i)) = \supp(P_{X|S}(\cdot|s_j)) = \X$, where $\X$ denotes the alphabet containing all possible values of $X$.
In this case, the sensitivity $f$ is as large as the pairwise distance in $\X$: $\triangle_f(\SetPair) = \triangle_{\X} = \max_{x,x' \in \X} d(x-x')$ and the resulting $\theta = \eta^{-1}(\epsilon/\triangle_{\X})$ could make the noise power too large to convey any useful information of $X$ in the disclosed dataset.
The following theorem proposes another approach based on a distance metric over the probability space.
The full proof is in Section~\ref{app:theo:W1Exponential} in the supplementary material

\begin{figure*}[ht]
%\vspace{.1in}
\centerline{
\subfigure[]{\includegraphics[scale=0.68]{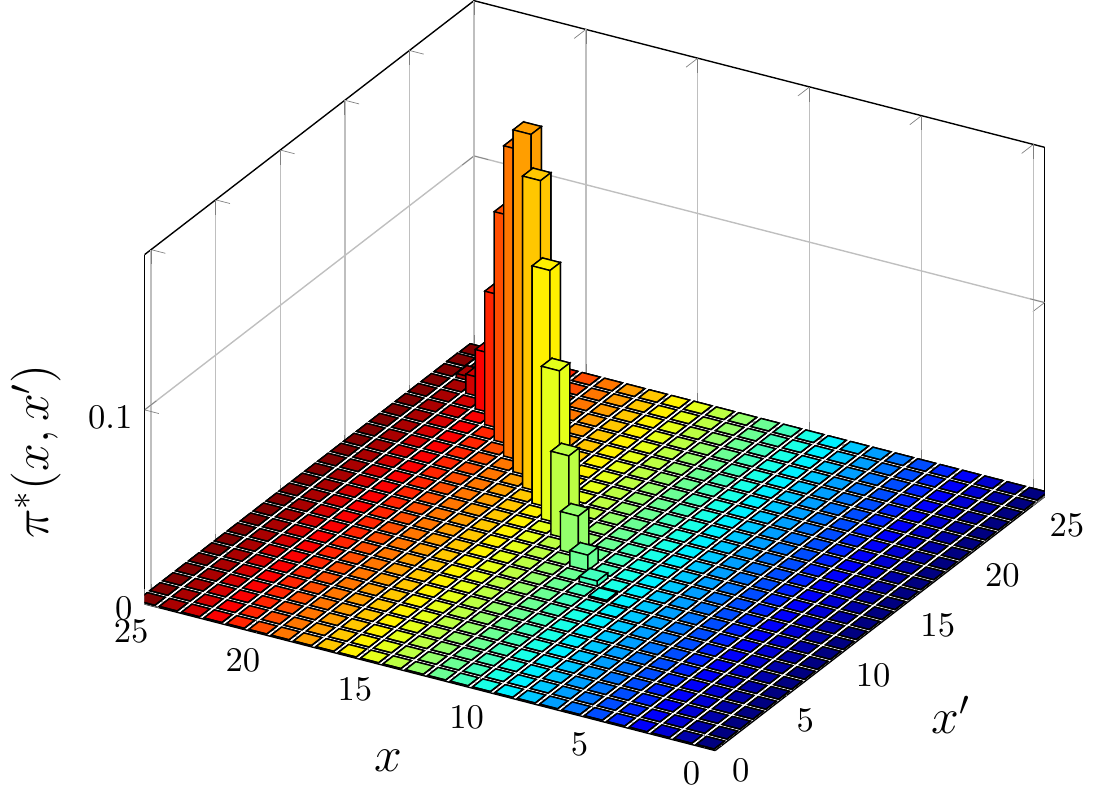}} \qquad
\subfigure[]{\includegraphics[scale=0.68]{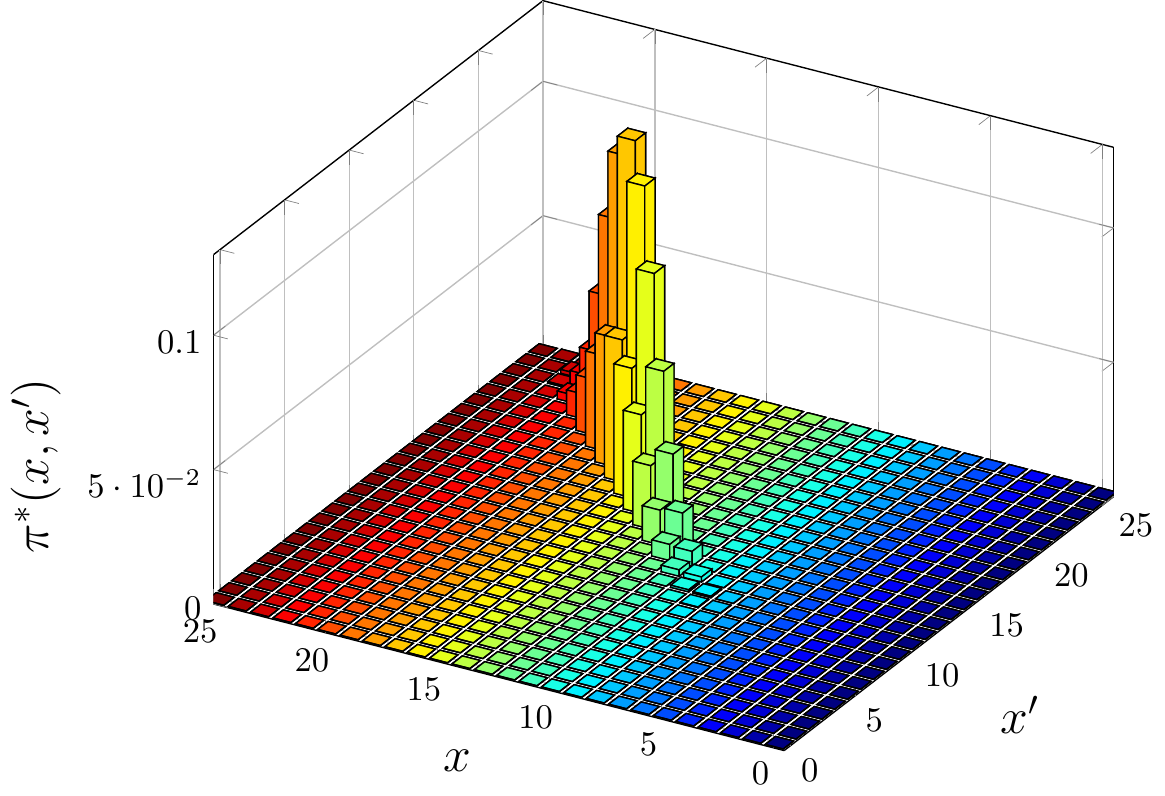}}
}
%\vspace{.05in}
\caption{The Kantorovich optimal transportation plan $\pi^*$ in the $V$ independent user system for the discriminative pair set $\SetPair = \Set{(S_i = 0, S_i=1) \colon i \in \N }$ in (a) and $\SetPair = \Set{(S_i = 0, S_i = \perp_i) \colon i \in \N}$ in (b). The function $f$ is a counting query $X = f(S) = \sum_{i \in \N} S_i$ and $S = (S_i \colon i \in \N) \sim \text{Binomial}(V,p)$. We set $V = 25$ and $p = 0.7$. } \label{fig:Binomial}
\end{figure*}

\begin{theorem}[Kantorovich-exponential mechanism] \label{theo:W1Exponential}
    For the exponential mechanism, adding noise $N_{\theta}$ with\footnote{In~\eqref{eq:W1ExponentialSym}, there is a Kantorovich optimal transport plan $\pi^*$ for each $\rho \in \D$ and $(s_i,s_j) \in \SetPair$. }
    \begin{equation} \label{eq:W1ExponentialSym}
        \theta = \max_{\rho \in \D, (s_i,s_j) \in \SetPair} \eta^{-1}  \Big(  \epsilon / \sup_{(x,x') \in \supp(\pi^*)} d(x-x')  \Big)
    \end{equation}
\end{theorem}
attains $(\epsilon,\SetPair)$-pufferfish privacy in $Y$.
\begin{proof}
For any pair $(s_i,s_j) \in \SetPair$, we have
\begin{align} \label{eq:W1Exponential:SuffCond1Main}
     &P_{Y|S}(y|s_i,\rho) - e^\epsilon P_{Y|S}(y|s_j,\rho) \nonumber \\
    % & = \int P_{N_{\theta}}(y-x) P_{X|S}(x|s_i,\rho) \dif x - e^{\epsilon} \int P_{N_{\theta}}(y-x') P_{X|S}(x|s_j,\rho) \dif x \nonumber \\
     & =  \int \big( P_{N_{\theta}}(y-x) - e^{\epsilon} P_{N_{\theta}}(y-x') \big) \dif \pi^*(x,x')  \nonumber\\
     & \leq \int P_{N_{\theta}}(y-x') \big( e^{\eta(\theta) d(x-x')}  - e^{\epsilon} \big)  \dif \pi^*(x,x'), \forall y.
\end{align}
For each $y$, \eqref{eq:W1ExponentialSym} is a sufficient condition for $e^{\eta(\theta) d(x-x')} \leq e^{\epsilon}$ for all $x,x'$, by which we have $\frac{P_{Y|S}(y|s_i,\rho)}{P_{Y|S}(y|s_j,\rho)} \leq e^{\epsilon}$.
Due to the symmetric property $d(x-x') = d(x'-x),\forall x,x'$, \eqref{eq:W1ExponentialSym} is also a sufficient condition for $\frac{P_{Y|S}(y|s_j,\rho)}{P_{Y|S}(y|s_i,\rho)} \leq e^{\epsilon}$.
\end{proof}

Theorem~\ref{theo:W1Exponential} essentially states that it is sufficient to only calibrating noise to the maximum pairwise distance over the support of the Kantorovich optimal transport plan $\pi^*$, which can be regarded as the \emph{sensitivity of $\pi^*$}.
It is clear that the maximum sensitivity of $\pi^*$ over all $(s_i,s_j) \in \SetPair$ is no greater than $\triangle_{f}(\SetPair)$. This has also been verified by~\citet[Theorem~3.3]{PufferfishWasserstein2017Song}.
In fact, in most cases, we have
$ \sup_{(x,x') \in \supp(\pi^*)} d(x-x') \ll \triangle_f(\SetPair) . $
See the experimental results in Section~\ref{sec:exp}.

\subsubsection{Interpretation of Lemma~\ref{lemma:NindUser}}
\label{sec:Lemma2Verify}

Theorem~\ref{theo:W1Exponential} in return explains Lemma~\ref{lemma:NindUser}.
For any deterministic query $f$ on the $V$ independent user system, we have the Kantorovich optimal transport plan $\pi^*(x,x') = 0$, for all $x,x'$ such that $d(x-x') > \triangle_f(\SetPair_i)$, i.e.,
\begin{equation} \label{eq:KantSupp}
    \supp(\pi^*) \subseteq \big\{(x,x') \colon d(x-x') \leq \triangle_f(\SetPair_i) \big\}.
\end{equation}
By Theorem~1, tuning $\theta$ to $\eta^{-1} \big( \epsilon / \triangle_f(\SetPair_i) \big)$ attains $(\epsilon,\SetPair_i)$-pufferfish privacy.

It is clear from \eqref{eq:KantSupp} that $\triangle_{f}(\SetPair_i)$ is in fact the sensitivity of the Kantorovich optimal transport plan $\pi^*$ for metric $d$. This is the key point that Lemma~2(a) is valid.

\paragraph{Separable query function}
For all separable query functions $f$ such that $X = f(S) = \sum_{i \in \N} f_i(S_i)$, we have
$$ f(s_i=a,s_{-i}) - f(s_i=b,s_{-i}) = f_i(a) - f_i(b), \forall s_{-i}. $$
In this case, $\pi^*(x,x') = 0$, for all $x,x'$ such that $x-x' \neq f_i(a) - f_i(b)$ and therefore
\begin{equation} \label{eq:KantSuppSep}
    \supp(\pi^*) = \big\{ (x,x') \colon x-x' = f_i(a) - f_i(b) \big\}.
\end{equation}
The sensitivity is $\triangle_f(\SetPair_i) = \max_{a,b\in \SA_i} d ( f_i(a) - f_i(b))$. $(\epsilon,\SetPair_i)$-pufferfish privacy attains by setting $\theta = \eta^{-1}  \big(  \epsilon / \max_{a,b\in \SA_i} d ( f_i(a) - f_i(b)) \big)$.

\begin{figure}[ht]
%\vspace{.1in}
\centerline{
\includegraphics[scale=0.75]{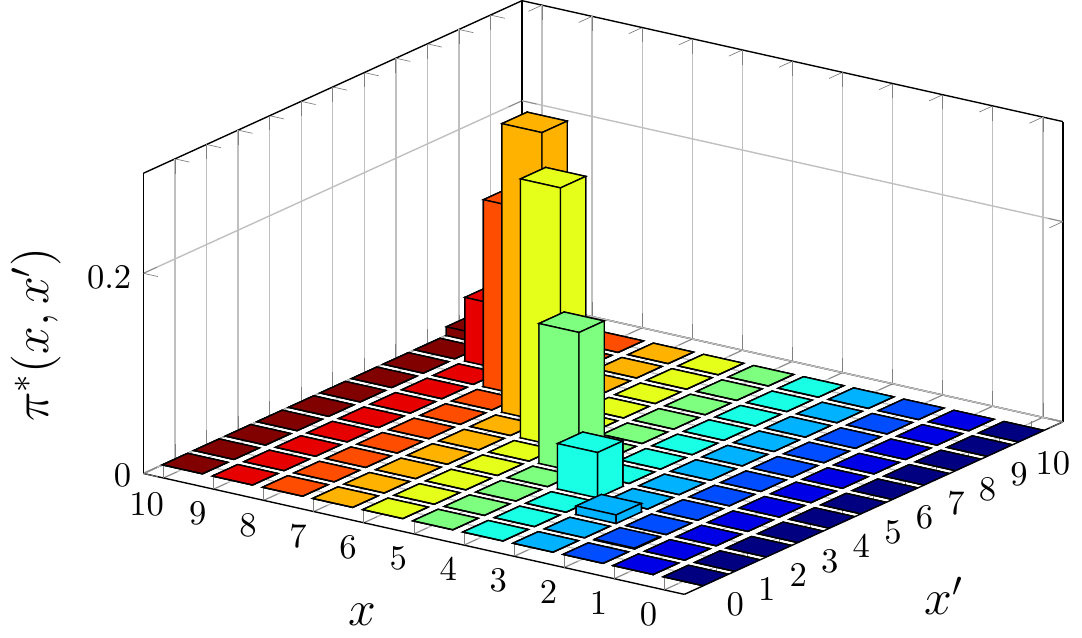}
%\subfigure[]{\includegraphics[scale=0.75]{figures/PoissonBinomialPiStar.pdf}} \qquad \qquad \quad
%\subfigure[]{\scalebox{0.7}{\input{figures/PoissonBinomialRelaxCompare.tex}}}
}
\caption{The Kantorovich optimal transportation plan $\pi^*$ in the $V$ independent user system for the discriminative pair $S_1 = 0$ and $S_1=1$ is in (a).
There are $V=10$ users.  For each user $i$, $S_i \sim \text{Bernoulli}(p_i)$ and so the counting query $X = \sum_{i \in \N} S_i$ follows Poisson Binomial distribution. }
\label{fig:PoissonBinomial}
\end{figure}

A typical example of separable query function is the counting query, where $X = f(S) = \sum_{i \in \N} S_i$.
In Figure~\ref{fig:Binomial}, we show the Kantorovich optimal transport plan $\pi^*$ for the counting query in a $25$ independent user system, where each $S_i \in \Set{0,1}$ follows $\text{Bernoulli}(0.7)$ distribution.
For discriminative secret pair $(S_i=0,S_i=1)$,
$\supp(\pi^*) = \big\{ (x,x') \colon x-x' = 1 \big\}$;
for discriminative pair $(S_i=0,S_i=\perp_i)$, $\supp(\pi^*) = \big\{ (x,x') \colon x-x' \leq 1 \big\}$.
In this case, adding noise $N_{\theta}$ with $\theta = \eta^{-1}  \big(  \epsilon / d(1) \big)$ attains $(\epsilon, \SetPair_i \cup \SetPair_{\perp_i})$-pufferfish privacy.

This example is equivalent to the single-prior privacy for answering counting query $f$ in \citet[Section~7.1.1]{Pufferfish2014Kifer}, which can be extended to the pufferfish private histogram publishing \citep[see][Section~7.1.2]{Pufferfish2014Kifer}.
The name `single prior' refers to the same probability distributions of $S_i$ for all $i$.
\citealp[Algorithm~1]{Pufferfish2014Kifer} proposes a segmented randomization scheme, in which the noise distribution in each segment needs to be determined by the prior knowledge $\rho$. Whereas, Lemma~\ref{lemma:NindUser} and Theorem~\ref{theo:W1Exponential} work for different $\rho$ and other query functions $f$.
For example, assuming each user determines his/her own coin flipping probability, i.e., $S_i \sim \text{Bernoulli}(\rho_i)$, we still have $\supp(\pi^*) = \big\{ (x,x') \colon x-x' = 1 \big\}$. See the Kantorovich optimal transport plan $\pi^*$ in Figure~\ref{fig:PoissonBinomial}.

See Section~\ref{lemmaInt:NindUser} in the supplementary materials for the full proof of Lemma~\ref{lemma:NindUser}(a) and (b) by Theorem~\ref{theo:W1Exponential} and the detailed derivation of \eqref{eq:KantSuppSep}.

\subsection{Relaxed Sufficient Condition}

The sufficient condition in Theorem~\ref{theo:W1Exponential} is strict in that it enforces the inequality $d(x-x') \leq  \epsilon$ to hold for each $(x,x') \in \supp(\pi^*)$ in \eqref{eq:W1Exponential:SuffCond1Main}.
However, the integral in \eqref{eq:W1Exponential:SuffCond1Main} is a noised expected distance. Knowing that the randomization does not increase statistical differences, Theorem~\ref{theo:W1Exponential} may result in a larger $\theta$ that overkill the data utility.\footnote{In \eqref{eq:W1Exponential:SuffCond1Main}, the distance $e^{\eta(\theta) d(x-x')} $ is averaged w.r.t. the joint probability $\pi^*(x,x')$ and then randomized by $P_{Y|X}(y|x) = P_{N_{\theta}}(y-x)$. Therefore, it is randomized statistical distance. }
We relax this sufficient condition in the following theorem. The full proof is in Section~\ref{app:theo:SuffCondRelax} in the supplementary material.

\begin{theorem}[relaxed sufficient condition]  \label{theo:SuffCondRelax}
    Let $\theta(s_i,s_j)$ be the maximum value of $\theta$ that holds the equalities
     \begin{subequations} \label{eq:SuffCondRelaxMain}
        \begin{align}
            & \int e^{\eta(\theta) d(x-x')} \pi^*(x,x') \dif x = e^{\epsilon} p(x'|s_j), \label{eq:SuffCondRelax1Main} \\
            & \int e^{\eta(\theta) d(x-x')} \pi^*(x,x') \dif x' = e^{\epsilon} p(x|s_i), \label{eq:SuffCondRelax2Main}
        \end{align}
    \end{subequations}
    over all $x$ and $x'$.
    For the exponential mechanism, adding noise $N_{\theta}$ with $\theta = \max_{\rho \in \D, (s_i,s_j) \in \SetPair} \theta(s_i,s_j)$ attains $(\epsilon,\SetPair)$-pufferfish privacy in $Y$.
\end{theorem}
\begin{proof}
Rewrite \eqref{eq:W1Exponential:SuffCond1Main} as
\begin{align*}
		&\int P_{N_{\theta}}(y-x') \big( e^{\eta(\theta) d(x-x')}  - e^{\epsilon} \big)  \dif \pi^*(x,x') \\
        %&= \iint P_{N_{\theta}}(y-x) \big( e^{\eta(\theta) d(x-x')} - e^{\epsilon} \big) \pi^*(x,x') \dif x \dif x' \\
        &= \int P_{N_{\theta}}(y-x')  \int \big( e^{\eta(\theta) d(x-x')} - e^{\epsilon}\big) \pi^*(x,x')  \dif x \dif x'.
\end{align*}
Relaxing the condition $e^{\eta(\theta) d(x-x')} \leq e^{\epsilon}, \forall x,x'$ to
\begin{multline} \label{eq:theo:SuffCondRelax}
		\int \big( e^{\eta(\theta) d(x-x')} - e^{\epsilon}\big) \pi^*(x,x')  \dif x  \leq 0 \\
        \Longrightarrow  \int e^{\eta(\theta) d(x-x')} \pi^*(x,x') \dif x \leq e^{\epsilon} p(x'|s_j), \forall x'
\end{multline}
still holds the inequality $\frac{P_{Y|S}(y|s_i,\rho)}{P_{Y|S}(y|s_j,\rho)} \leq e^{\epsilon}$. For $\eta$ nonincreasing in $\theta$, the minimum $\theta$ for the condition~\eqref{eq:theo:SuffCondRelax} is the one that holds~\eqref{eq:theo:SuffCondRelax} as an equality. We have \eqref{eq:SuffCondRelax1Main}. It can be shown in the same way that $\int e^{\eta(\theta) d(x-x')} \pi^*(x,x') \dif x' \leq e^{\epsilon} p(x|s_j), \forall x$ is a relaxed sufficient condition for $\frac{P_{Y|S}(y|s_j,\rho)}{P_{Y|S}(y|s_i,\rho)} \leq e^{\epsilon}$ and therefore \eqref{eq:SuffCondRelax2Main}. Maximize $\theta$ that holds \eqref{eq:SuffCondRelax1Main} and \eqref{eq:SuffCondRelax2Main} over $x'$ and $x$, respectively, and over all $(s_i,s_j) \in \SetPair$, we have Theorem~\ref{theo:SuffCondRelax}.
\end{proof}

The relaxation in Theorem~\ref{theo:SuffCondRelax} produces an $(\epsilon,\SetPair)$-pufferfish privacy achieving $\theta$ that is smaller than the one in Theorem~\ref{theo:W1Exponential}.\footnote{For two functions $f_1(\theta)$ and $f_2(\theta)$ both nonincreasing in $\theta$, $f_1^{-1}(a) \leq f_2^{-1}(a), \forall a$. }
Though for continuous $X$ solving the integral equations in~\eqref{eq:SuffCondRelaxMain} could be complex, it is convenient to apply Theorem~\ref{theo:SuffCondRelax} to the integer-valued metric $d$, where \eqref{eq:SuffCondRelaxMain} reduces to
\begin{subequations} \label{eq:SuffCondRelaxpPoly}
\begin{align}
    & \sum_{x,x'} e^{\eta(\theta) d(x-x')} \pi^*(x,x') = e^{\epsilon} p(x'|s_j),\quad  \forall x'\\
    & \sum_{x,x'} e^{\eta(\theta) d(x-x')} \pi^*(x,x') = e^{\epsilon} p(x|s_i), \quad \forall x
\end{align}
\end{subequations}
and $\theta(s_i,s_j)$ for each $(s_i,s_j) \in \SetPair$ can be determined by solving the polynomial equations above.

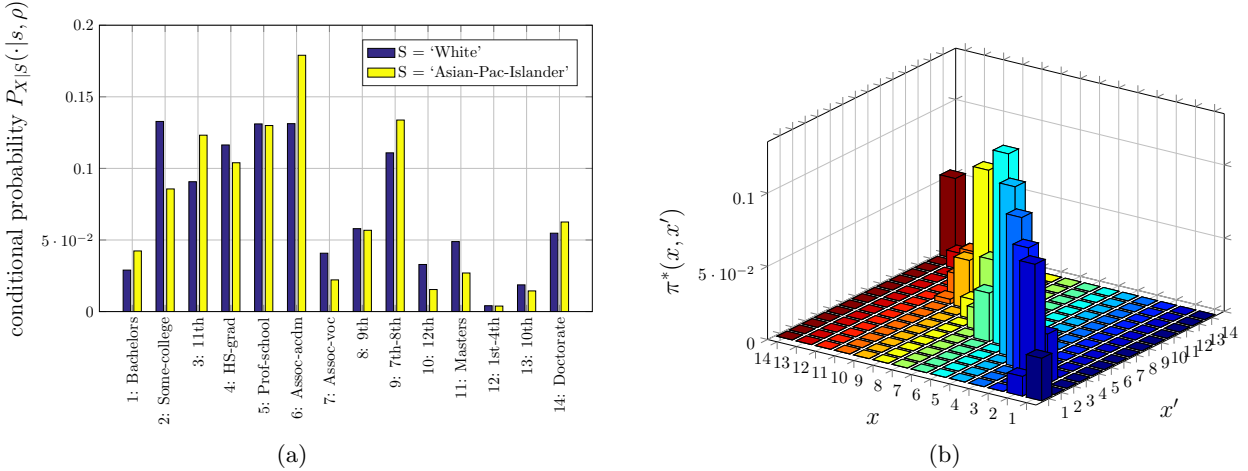
\begin{figure*}[ht]
%\vspace{.1in}
\centerline{
\subfigure[]{\scalebox{0.6}{% This file was created by matlab2tikz.
\definecolor{mycolor1}{rgb}{0.20810,0.16630,0.52920}%
\definecolor{mycolor2}{rgb}{0.97630,0.98310,0.05380}%
\begin{tikzpicture}

\begin{axis}[%
width=4.3in,
height=2.5in,
scale only axis,
bar shift auto,
xmin=0,
xmax=15,
xtick={1,2,3,4,5,6,7,8,9,10,11,12,13,14},
xticklabels={{1: Bachelors},{2: Some-college},{3: 11th},{4: HS-grad},{5: Prof-school},{6: Assoc-acdm},{7: Assoc-voc},{8: 9th},{9: 7th-8th},{10: 12th},{11: Masters},{12: 1st-4th},{13: 10th},{14: Doctorate},{15: 5th-6th},{16: Preschool}},
xticklabel style={rotate=90},
ymin=0,
ymax=0.2,
ylabel={\Large conditional probability $P_{X|S}(\cdot|s,\rho)$},
grid=major,
legend style={at={(0.97,0.95)},draw=darkgray!60!black,fill=white,legend cell align=left}
]
\addplot[ybar, bar width=0.229, fill=mycolor1, draw=black, area legend] table[row sep=crcr] {%
1	0.0289750871769062\\
2	0.132796491354208\\
3	0.0906999316964446\\
4	0.116367688823381\\
5	0.131070927849876\\
6	0.131250674048244\\
7	0.0407664377898408\\
8	0.0579142251141388\\
9	0.110867455153324\\
10	0.0328935543013265\\
11	0.04889096595607\\
12	0.00409821332278822\\
13	0.0186576553905885\\
14	0.0547506920228637\\
};
%\addplot[forget plot, color=white!15!black] table[row sep=crcr] {%
%0	0\\
%15	0\\
%};
\addlegendentry{S = `White'}

\addplot[ybar, bar width=0.229, fill=mycolor2, draw=black, area legend] table[row sep=crcr] {%
1	0.0423484119345525\\
2	0.0856592877767084\\
3	0.123195380173244\\
4	0.103946102021174\\
5	0.129932627526468\\
6	0.179018286814244\\
7	0.0221366698748797\\
8	0.0567853705486044\\
9	0.133782483156882\\
10	0.0153994225216554\\
11	0.026948989412897\\
12	0.00384985563041386\\
13	0.014436958614052\\
14	0.0625601539942252\\
};
%\addplot[forget plot, color=white!15!black] table[row sep=crcr] {%
%0	0\\
%15	0\\
%};
\addlegendentry{S = `Asian-Pac-Islander'}

\end{axis}
\end{tikzpicture}% 
}} \qquad
\subfigure[]{\includegraphics[scale=0.65]{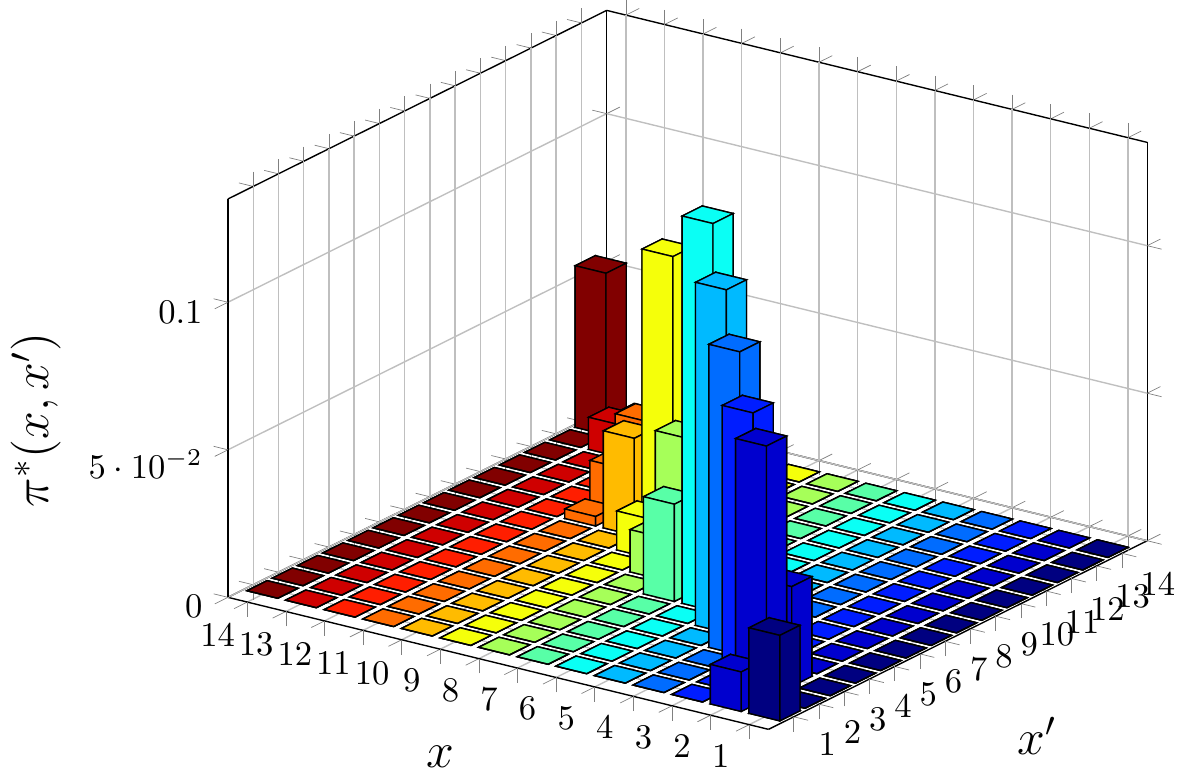}}
}
%\vspace{.05in}
\caption{For $S$ and $X$ being the attributes `race' and `education', respectively, in the adult database \citep{UCI2007}, (a) shows the distribution of `education' conditioned on two events: the `race' is `White' and the `race' is `Asian-Pac-Islander'. (b) shows the corresponding Kantorovich optimal transport plan $\pi^*$. } \label{fig:AdultPiStar}
\end{figure*}

\section{EXPERIMENT}
\label{sec:exp}
In the UCI machine learning repository \citep{UCI2007}, the adult dataset was extracted from the census bureau database to predict whether the participant's income exceeds 50K/yr. It contains $32652$ records and $15$ attributes.
In this experiment, we use 2 attributes/columns. Let $S$ and $X$ denote the columns `race' and `education', respectively.
That is, we want to publish the column `education' while protecting the privacy of `race' for all participants.

We assign a numeric index to $X$: value $1$ denotes `Bachelors', value $2$ denotes `Some-college' and so on. See horizontal axis in Figure~\ref{fig:AdultPiStar}(a) for the alphabet $\X$ containing all possible values of `education'.
Consider the events $S = \text{`White'}$ and $S = \text{`Asian-Pac-Islander'}$. The probability of `education' $X$ conditioned on each event is plotted in Figure~\ref{fig:AdultPiStar}(a).
The support of both $P_{X|S}(\cdot|\text{`White'},\rho)$ and $P_{X|S}(\cdot|\text{`Asian-Pac-Islander'},\rho)$ is $\X = \Set{1,\dotsc,14}$.
In Figure~\ref{fig:AdultPiStar}(a), we show the corresponding Kantorovich optimal transport plan $\pi^*$.
Take the Laplace noise for example. We have
$\ell_1$-norm $d(z) = |z|$ and
\begin{align*}
     &\triangle_{\X} = \max_{x,x' \in \X} |x-x'| = 14, \\
     &\max_{(x,x') \in \supp(\pi^*)} |x-x'| = 2.
\end{align*}
The noise power is much reduced if we calibrating $\theta$ to the sensitivity of $\supp(\pi^*)$ rather than $\triangle_{\X}$.

\begin{figure}[t]
%\vspace{.1in}
\centerline{\scalebox{0.6}{% This file was created by matlab2tikz.
\begin{tikzpicture}

\begin{axis}[%
width=4.5in,
height=2.5in,
scale only axis,
xmin=0.8,
xmax=5.8,
xlabel={\Large privacy budget $\epsilon$},
ymin=0,
ymax=12.8,
ylabel={\Large utility loss $\VAR[N_{\theta}]$},
grid=major,
legend style={at={(1.08,0.93)},draw=darkgray!60!black,fill=white,legend cell align=left}
]

\addplot [
color=blue,
line width = 1.5pt]
table[row sep=crcr]{%
0.8	12.5\\
1.3	4.73372781065089\\
1.8	2.46913580246914\\
2.3	1.51228733459357\\
2.8	1.02040816326531\\
3.3	0.734618916437098\\
3.8	0.554016620498615\\
4.3	0.432666306111412\\
4.8	0.347222222222222\\
5.3	0.284798860804557\\
5.8	0.237812128418549\\
};
\addlegendentry{\large Laplace noise by sufficient condition Theorem~\ref{theo:W1Exponential}};

\addplot [
color=red,
dashed,
line width = 1.5pt]
  table[row sep=crcr]{%
0.8	3.125\\
1.3	1.18343195266272\\
1.8	0.617283950617284\\
2.3	0.397579269785382\\
2.8	0.305394110969956\\
3.3	0.244884060038036\\
3.8	0.202304351924927\\
4.3	0.170824401238967\\
4.8	0.146680137698887\\
5.3	0.127631329335633\\
5.8	0.112262755234664\\
};
\addlegendentry{\large Laplace noise by relaxed sufficient condition Theorem~\ref{theo:SuffCondRelax}};

%\addplot [
%color=red,
%line width = 1.5pt]
%  table[row sep=crcr]{%
%0.8	2.5\\
%1.3	1.53846153846154\\
%1.8	1.11111111111111\\
%2.3	0.869565217391304\\
%2.8	0.714285714285714\\
%3.3	0.606060606060606\\
%3.8	0.526315789473684\\
%4.3	0.465116279069767\\
%4.8	0.416666666666667\\
%5.3	0.377358490566038\\
%5.8	0.344827586206897\\
%};
%\addlegendentry{\large Kantorovich-Gaussian by Theorem~\ref{theo:W1Exponential}};
%
%\addplot [
%color=red,
%dashed,
%line width = 1.5pt]
%  table[row sep=crcr]{%
%0.8	0.753735108152688\\
%1.3	0.569426588248364\\
%1.8	0.475329419379081\\
%2.3	0.414685025826318\\
%2.8	0.370769811705107\\
%3.3	0.336733742389958\\
%3.8	0.309184339385332\\
%4.3	0.28621552555142\\
%4.8	0.266654470002208\\
%5.3	0.249728201137357\\
%5.8	0.234899332087678\\
%};
%\addlegendentry{\large Kantorovich-Gaussian by Theorem~\ref{theo:SuffCondRelax}};

\end{axis}
\end{tikzpicture}% 
}}
%\vspace{.1in}
\caption{The variance of Laplace noise then Theorem~\ref{theo:W1Exponential} and Theorem~\ref{theo:SuffCondRelax} are applied to the optimal transportation plan $\pi^*$ in Figure~\ref{fig:AdultPiStar}(b). } \label{fig:AdultRelaxCompare}
\end{figure}
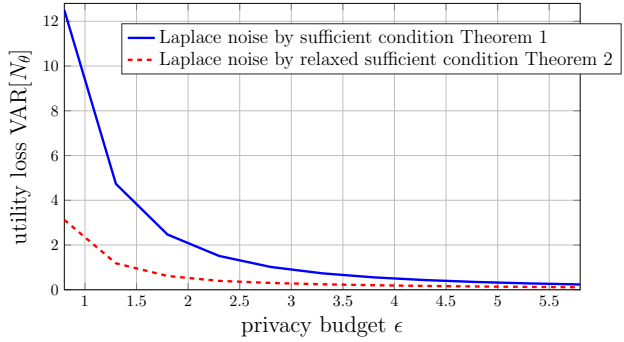

Increasing the privacy budget $\epsilon$ from $0.8$, we apply the sufficient condition in Theorem~\ref{theo:W1Exponential} and the relaxed sufficient condition in Theorem~\ref{theo:SuffCondRelax} to obtain $\theta$ for Laplace noise.
Figure~\ref{fig:AdultRelaxCompare} shows the noise variance as a function of $\epsilon$.
Here, the value of $\theta$ in Theorem~\ref{theo:SuffCondRelax} is obtained by solving the polynomial equations in \eqref{eq:SuffCondRelaxpPoly}.
It can be seen that the resulting noise power by the relaxed sufficient condition in Theorem~\ref{theo:SuffCondRelax} is always less than Theorem~\ref{theo:W1Exponential}.

%The same as Corollary~\ref{coro:BeyondLaplace}, when the privacy budget $\epsilon$ is small,
%%
%Kantorovich-Gaussian mechanism greatly improves data utility, as compared to the Kantorovich-Laplace mechanism, when $\epsilon$ is small, e.g., by Theorem~\ref{theo:SuffCondRelax}, the variance of Laplace noise is about $4.15$ times larger than the Gaussian noise when $\epsilon = 0.8$.

\section{$\delta$-Approximation by Gaussian}

\citet{RDP2017} pointed out two reasons that Gaussian noise is preferred over Laplace noise: the noise variance proportional to the $\ell_2$-norm is no larger than $\ell_1$-norm; the tail probability of Gaussian distribution decays faster than Laplace distribution.
This section considers zero-mean Gaussian noise $N_{\theta}$: $P_{N_{\theta}}(z) = \frac{1}{\sqrt{2\pi} \theta} e^{-\frac{z^2}{2\theta^2}}$.

It is shown in \citet[Theorem~3.22]{Dwork2014book} that for $\delta \in (0,1)$, Gaussian mechanism attains $\delta$-approximation of differential privacy: $\left| \log \frac{P_{Y|S}(y|s_i)}{P_{Y|S}(y|s_j)} \right| \leq \epsilon$ for all neighboring $s_i$ and $s_j$ with probability at least $1-\delta$.
Following the definition of $(\epsilon,\delta)$-differential privacy, we say that $Y$ is a $\delta$-approximation of $(\epsilon,\SetPair)$-pufferfish privacy if
$$ \log \frac{P_{Y|S}(y|s_i, \rho )-\delta}{P_{Y|S}(y|s_j, \rho)}  \leq \epsilon \text{, } \log \frac{P_{Y|S}(y|s_j, \rho )-\delta}{P_{Y|S}(y|s_i, \rho)}  \leq \epsilon$$
for all $(s_i,s_j) \in \SetPair$ and $\rho \in \D$.
The theorem below states that to achieve this approximation, it suffices to calibrate the standard deviation $\theta$ of Gaussian noise to the sensitivity of the Kantorovich optimal transport plan.

\begin{theorem} \label{theo:GaussN_AISTATS}
    For $\delta \in (0,1)$ and $N_{\theta}$ being zero-mean Gaussian noise, $Y$ attains $\delta$-approximation of $\epsilon$-pufferfish privacy
    \begin{enumerate}[(a)]
        \item if $\theta \geq \frac{\sqrt{2 \log (1.25 /\delta)}}{\epsilon} \triangle$ for all $\epsilon \leq 1$;
        \item if $\theta = \frac{\triangle}{\epsilon} c$ for $c > 0.41 \delta^{-\frac{1}{3}} + \sqrt{ (0.41\delta^{-\frac{1}{3}})^2 + \frac{\epsilon}{2} }$,
    \end{enumerate}
    where $\triangle = \sup_{(x,x') \in \supp(\pi^*)} |x-x'|$ is the $\ell_1$-sensitivity of the Kantorovich optimal transport plan $\pi^*$.
\end{theorem}
\begin{proof}
Let $\theta = \frac{\triangle}{\epsilon} c$.
To have
\begin{align*}
    & P_{Y|S}(y|s_i,\rho) - e^\epsilon P_{Y|S}(y|s_j,\rho) \\
    &= \int \frac{1}{\sqrt{2\pi} \theta} \left( e^{-\frac{(y-x)^2}{2\theta^2}}  - e^{\epsilon-\frac{(y-x')^2}{2\theta^2}} \right)  \dif \pi^*(x,x') \leq 0,
\end{align*}
it suffice to make $\frac{(y-x')^2-(y-x)^2}{2\theta^2} \leq \frac{2\triangle|y-x| + \triangle^2}{2\theta^2} \leq \epsilon \Longrightarrow \frac{|y-x|}{\theta} \leq c - \frac{\epsilon}{2c}$, where $\frac{|Y-X|}{\theta} = \frac{N_{\theta}}{\theta}$ is standard normally distributed.
To have $\Pr(\frac{|y-x|}{\theta} \leq t) \geq 1-\delta$, we use the Gaussian tail bound $\Pr(\frac{y-x}{\theta} > t) < \frac{1}{\sqrt{2 \pi } t} e^{-\frac{t^2}{2}} < \frac{\delta}{2}$, which can be written as
\begin{equation} \label{eq:theo:GaussN_AISTATS}
    \log t + \frac{t^2}{2} > \log \sqrt{\frac{2}{\pi}} \frac{1}{\delta}.
\end{equation}
Substituting $t = c - \frac{\epsilon}{2c}$, we have $c^2 >2 \log (1.25 / \delta)$ and (a).
To prove (b), we relieve the constraint $\epsilon \leq 1$ and apply the inequality $t-1 \geq \log t, \forall t > 0$ to \eqref{eq:theo:GaussN_AISTATS}, i.e., request
$\log t + \log \frac{t^2}{2} + 1 \geq \log t + \log \frac{t^2}{2}  > \log \sqrt{\frac{2}{\pi}} \frac{1}{\delta}$.
We have $t = c - \frac{\epsilon}{2c} > \left( \frac{2}{e} \right)^{\frac{1}{3}} \left( \frac{2}{\pi}\right)^{\frac{1}{6}} \delta^{-\frac{1}{3}}$, where $\left( \frac{2}{e} \right)^{\frac{1}{3}} \left( \frac{2}{\pi}\right)^{\frac{1}{6}} = 0.8373$.
Solving the quadratic inequality $c - \frac{\epsilon}{2c} > 0.84 \delta^{-\frac{1}{3}}$, we get $c > 0.41 \delta^{-\frac{1}{3}} + \sqrt{(0.41 \delta^{-\frac{1}{3}})^2 + \frac{\epsilon}{2}}$.
See Section~\ref{app:theo:GaussN_AISTATS} in the supplementary material for the full proof of Theorem~\ref{theo:GaussN_AISTATS}.
\end{proof}

The proof of Theorem~\ref{theo:GaussN_AISTATS}(a) is similar to \citet[Appendix~A]{Dwork2014book} for $(\epsilon,\delta)$-differential privacy, except that the noise should be calibrated to the sensitivity of the Kantorovich optimal transport plan $\pi^*$, instead of the query function $f$.

\begin{remark} \label{rem:parallel}

Lemma~\ref{lemma:Winfty2W1Laplace}, Theorem~\ref{theo:W1Exponential} and Theorem~\ref{theo:GaussN_AISTATS}(a) parallel the well-known results on Laplace, exponential and Gaussian mechanisms for differential privacy: replacing the sensitivity of the query function $f$ in Theorems~3.6, 3.10 and 3.22 in \cite{Dwork2014book} by the maximum pairwise distance in the Kantorovich optimal transport plan $\pi^*$ over all $(s_i,s_j) \in \SetPair$ and $\rho \in \D$, the additive noise mechanism attains pufferfish privacy.

\end{remark}

\section{CONCLUSION}

We studied the problem of how to attain pufferfish privacy, the $\epsilon$-indistinguishability when the secret $S$ is correlated with the public data $X$, by adding independent noise $N$ to $X$.
We proved that calibrating noise to the maximum pairwise distance over the support of the Kantorovich optimal transport plan $\pi^*$ attains pufferfish privacy.
Unlike the difficulty in determining the optimal transport plan in the existing $\infty$-Wasserstein mechanism, $\pi^*$ is directly obtained by the conditional probabilities $P_{X|S}(\cdot|s_i,\rho)$ and $P_{X|S}(\cdot|s_j,\rho)$ for every secret pair $(s_i,s_j)$.
We also derived a relaxed sufficient condition and showed that it enhances data utility for integer-valued $X$.

This paper in fact proposes a method for attaining the pufferfish privacy based on the  mass transport cost upper bound $C(x,x';\theta)$: for any noise distribution $P_{N_{\theta}}(\cdot)$ such that $P_{N_{\theta}}(y-x) \leq C(x,x';\theta)  P_{N_{\theta}}(y-x'), \forall x,x',y$, pufferfish privacy attains if $\sup_{(x,x') \in \supp(\pi^*)} C(x,x';\theta) \leq e^{\epsilon}$, for all $(s_i,s_j) \in \pi^*$ and $\rho \in \D$. See \eqref{eq:W1Exponential:SuffCond1Main}.
Here, $C(x,x';\theta)$ does not need to be an exponential function. Therefore, it is worth exploring noise distributions other than the exponential mechanism.
For $P_{X|S}(\cdot|s,\rho)$ being Gaussian distribution or Gaussian mixture model for all $s$, the Kantorovich optimal transport plan $\pi^*$ is fully characterized by the mean and covariance matrix~\citep{Takatsu2010WGauss,Delon2020WGauss}.
Since Gaussian models are widely used in machine learning, it is of interest whether the Kantorovich mechanism can be apply to the privacy-preserving pattern recognition problems.

\subsubsection*{Acknowledgements}
The author would like to thank A/Prof Olya Ohrimenko for helping her initiate the study on pufferfish privacy and Prof Ben Rubinstein for his useful advice on the dissemination of the research results in this paper.

\bibliographystyle{apalike}
\bibliography{BIB}

%%%%%%%%%%%%%%%%%%%%%%%%%%%%%%%%%%%
%%%%%% SUPPLEMENT (OPTIONAL) %%%%%%
%%%%%%%%%%%%%%%%%%%%%%%%%%%%%%%%%%%

\clearpage
\appendix

\thispagestyle{empty}

% For one-column format, uncomment the following:
\onecolumn \makesupplementtitle
% For two-column format, uncomment the following:
%\twocolumn[ \makesupplementtitle ]

\section{$V$ INDEPENDENT USER SYSTEM}
\label{appA:lemma:NindUser}

In the $V$ independent user system, Let $\rho$ be the prior knowledge about the probability distribution  $P_{S_i}(\cdot)$ for all $i \in \N$.
Since $S_i$'s are independent random variables, $P_S(s) = \Pi_{i \in \N} P_{S_i}(s_i)$ and $P_{S|S_i}(s|s_i) = P_{S_{-i}}(s_{-i}), \forall i \in \N$, where $S_{-i} = (S_j \colon j \in \N \setminus \Set{i})$ denote the multiple random variable excluding dimension $i$.
We also have the conditional probabilities
\begin{align*}
    P_{X|S_i}(x |a, \rho) &=  \Pr(f(S)=x | S_i = a) \\
    & = \int_{\SA_{-i}(x,a)} P_{S|S_i}(s|a,\rho) \dif s_{-i} \\
    & = \int_{\SA_{-i}(x,a)}  P_{S_{-i}} (s_{-i}) \dif s_{-i}\\
    & = \Pr(f(S_i = a, S_{-i}) = x), \\
     P_{X|S_i}(x | \perp_i, \rho) &= \Pr(f(S)=x) \\
     &= \int_{\SA(x)} P_{S|S_i}(s|\perp_i) \dif s \\
     &= \int_{\SA(x)}  P_{S} (s) \dif s \\
     &= \int \Big( \int_{\SA_{-i}(x,s_i)} P_{S_{-i}}(s_{-i})  \dif s_{-i} \Big) P_{S_i}(s_i) \dif s_i \\
     &= \int P_{X|S_i}(x | s_i,\rho ) P_{S_i}(s_i) \dif s_i,
\end{align*}
where $\SA(x) = \Set{s \colon x = f(s)}$ and $\SA_{-i}(x,a) = \Set{s_{-i} \colon x = f(s_i = a, s_{-i})}$. The last equality above means $\Pr(f(S)=x) = \E_{S_i} [\Pr(f(S)=x | S_i)]$.

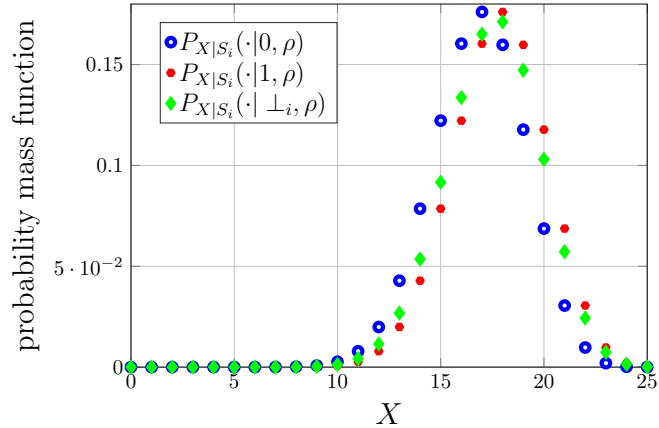
\begin{figure}[h]
%\vspace{.3in}
\centerline{\scalebox{0.8}{% This file was created by matlab2tikz v0.4.3.
% Copyright (c) 2008--2013, Nico Schlömer <nico.schloemer@gmail.com>
% All rights reserved.
%
% The latest updates can be retrieved from
%   http://www.mathworks.com/matlabcentral/fileexchange/22022-matlab2tikz
% where you can also make suggestions and rate matlab2tikz.
%
\begin{tikzpicture}

\begin{axis}[%
width=4in,
height=3in,
xmin=0,
xmax=25,
xlabel={\Large $X$},
ymin=0,
ymax=0.18,
ylabel={\Large probability mass function},
grid=major,
legend style={at={(0.4,0.95)},draw=darkgray!60!black,fill=white,legend cell align=left}
]
\addplot [
color=blue,
only marks,
line width = 2pt,
mark=o,
mark options={solid},
]
table[row sep=crcr]{
0 2.82429536481001e-13\\
1 1.58160540429361e-11\\
2 4.24397450152117e-10\\
3 7.26191192482512e-09\\
4 8.89584210791077e-08\\
5 8.30278596738338e-07\\
6 6.13483629812216e-06\\
7 3.6809017788733e-05\\
8 0.000182511379869134\\
9 0.000757084242420112\\
10 0.00264979484847039\\
11 0.00786908773182116\\
12 0.0198913050998813\\
13 0.0428428109843597\\
14 0.078545153471326\\
15 0.122181349844285\\
16 0.160363021670624\\
17 0.176084886540293\\
18 0.159780730379155\\
19 0.117733169753061\\
20 0.0686776823559524\\
21 0.0305234143804233\\
22 0.00971199548468013\\
23 0.00197054980848582\\
24 0.000191581231380566\\
25 0\\
};
\addlegendentry{\large $P_{X|S_i}(\cdot|0,\rho)$};

\addplot [
color=red,
only marks,
line width = 2pt,
mark=asterisk,
mark options={solid},
]
table[row sep=crcr]{
0 0\\
1 2.82429536481001e-13\\
2 1.58160540429361e-11\\
3 4.24397450152117e-10\\
4 7.26191192482512e-09\\
5 8.89584210791077e-08\\
6 8.30278596738338e-07\\
7 6.13483629812216e-06\\
8 3.6809017788733e-05\\
9 0.000182511379869134\\
10 0.000757084242420112\\
11 0.00264979484847039\\
12 0.00786908773182116\\
13 0.0198913050998813\\
14 0.0428428109843597\\
15 0.078545153471326\\
16 0.122181349844285\\
17 0.160363021670624\\
18 0.176084886540293\\
19 0.159780730379155\\
20 0.117733169753061\\
21 0.0686776823559524\\
22 0.0305234143804233\\
23 0.00971199548468013\\
24 0.00197054980848582\\
25 0.000191581231380566\\
};
\addlegendentry{\large $P_{X|S_i}(\cdot|1,\rho)$};

\addplot [
color=green,
only marks,
line width = 2pt,
mark=diamond,
mark options={solid},
]
table[row sep=crcr]{
0 8.47288609443003e-14\\
1 4.94251688841752e-12\\
2 1.3839047287569e-10\\
3 2.47565179255402e-09\\
4 3.17708646711099e-08\\
5 3.11354473776877e-07\\
6 2.42164590715349e-06\\
7 1.53370907453054e-05\\
8 8.05197264128534e-05\\
9 0.000354883238634428\\
10 0.0013248974242352\\
11 0.00421558271347562\\
12 0.0114757529422392\\
13 0.0267767568652248\\
14 0.0535535137304496\\
15 0.0916360123832137\\
16 0.133635851392187\\
17 0.165079581131525\\
18 0.171193639691951\\
19 0.147166462191327\\
20 0.103016523533929\\
21 0.0572314019632936\\
22 0.0242799887117003\\
23 0.00738956178182184\\
24 0.00143685923535425\\
25 0.000134106861966396\\
};
\addlegendentry{\large $P_{X|S_i}(\cdot|\perp_i,\rho)$};

\end{axis}
\end{tikzpicture}% 
}}
%\vspace{.3in}
\caption{The conditional probability $P_{X|S_i}(x|s_i,\rho)$ for the events $S_i = 0$, $S_i = 1$ and $S_i = \perp_i$ in the $V$ independent user system. The number of users is $V=25$. Each $S_i \sim \text{Bernoulli}(0.7)$. The function $f$ is a counting query $X = f(S) = \sum_{i \in \N} S_i$ and $X \sim \text{Binomial}(25,0.7)$.} \label{fig:PMF}
\end{figure}

In Figure~\ref{fig:PMF}, we show the conditional probability mass function $P_{X|S_i}(x |\cdot, \rho)$ in an $V=25$ independent user system, where $\SA_i \in \Set{0,1,\perp_i}$. Each dimension $S_i$ in the multiple random variable $S = (S_i \colon i \in \N)$ follows $\text{Bernoulli}(p)$ distribution, where $ p= 0.7$. %\footnote{Parameter $p=0.7$ denotes the probability when $S_i = 1$ given that $S_i \neq \perp_i$. Let $P_{S_i}(\perp_i) = \zeta_i$.
%We have the probabilities $P_{S_i}(0) = (1-p) (1-\zeta_i)$ and $P_{S_i}(1) = p (1-\zeta_i)$. But, the conditional probability $P_{X|S_i}(x|s_i,\rho)$ in Figure~1 is independent of $\zeta_i$.}
%
The corresponding Kantorovich optimal transport plan $\pi^*$ is shown in Figure~1 in the main submission.

For $Y = X + N$, where $N$ is independent of $X$,
\begin{align*}
    P_{Y|S_i}(y|a,\rho) &= \int P_{N}(y-x) P_{X|S_i}(x|a,\rho) \dif x \\
    & = \int P_{N}(y - f(s_i=a,s_{-i})) P_{S_{-i}}(s_{-i}) \dif s_{-i}, \\
    P_{Y|S_i}(y|\perp_i,\rho) &= \int P_{N}(y-x) P_{X|S_i}(x|\perp_i,\rho) \dif x \\
    & = \int P_{N}(y - f(s)) P_{S}(s) \dif s \\
    & = \iint P_{N}(y - f(s_i,s_{-i})) P_{S_{-i}}(s_{-i}) P_{S_i}(s_i) \dif s_{-i} \dif s_i\\
    &= \int P_{Y|S_i}(y | s_i,\rho) P_{S_i}(s_i) \dif s_i.
\end{align*}
The last equality above means $P_{Y|S_i}(y|\perp_i,\rho) = \E_{S_i} [P_{Y|S_i}(y|\cdot,\rho)]$.

\section{PROOF of LEMMA~\ref{lemma:Winfty2W1Laplace}}
\label{app:lemma:Winfty2W1Laplace}

Denote $[z]_+= \max\Set{z,0}$ for all $z \in \Real$.
First, for all $\pi \in \Gamma(s_i,s_j)$,
$\inf_{\pi \in \Gamma(s_i,s_j)} \sup_{(x,x') \in \supp(\pi)} \Big( \frac{|x-x'|}{\theta} - \epsilon \Big) = 0 $ is equivalent to
$$\inf_{\pi \in \Gamma(s_i,s_j)} \int \Big[ \frac{|x-x'|}{\theta} - \epsilon \Big]_+ \dif \pi(x,x') = 0.$$
Then, for the $\infty$-Wasserstein mechanism with $\theta = \frac{1}{\epsilon}\max_{\rho \in \D,(s_i,s_j) \in \SetPair}  W_\infty(s_i, s_j)$ proposed by \citet{PufferfishWasserstein2017Song}, we have
\begin{align}
    \theta &= \frac{1}{\epsilon}\max_{\rho \in \D,(s_i,s_j) \in \SetPair}  W_\infty(s_i, s_j) \\
    &= \frac{1}{\epsilon}\max_{\rho \in \D, (s_i,s_j) \in \SetPair}  \inf_{\pi \in \Gamma(s_i,s_j)} \sup_{(x,x') \in \supp(\pi)} |x-x'| \nonumber \\
    &= \max\limits_{\rho \in \D, (s_i,s_j) \in \SetPair} \Big\{ \theta \colon \inf\limits_{\pi \in \Gamma(s_i,s_j)} \sup\limits_{(x,x') \in \supp(\pi)}  \left( \frac{|x-x'|}{\theta} - \epsilon \right) = 0 \Big\}  \nonumber \\
    & = \max\limits_{\rho\in\D,(s_i,s_j) \in \SetPair} \Big\{ \theta \colon \inf\limits_{\pi \in \Gamma(s_i,s_j)} \int \Big[ \frac{|x-x'|}{\theta} - \epsilon \Big]_+ \dif \pi(x,x') = 0 \big\}\label{eq:lemma:Winfty2W1Aux1} \\
    & =\max\limits_{\rho \in \D, (s_i,s_j) \in \SetPair} \Big\{ \theta \colon \int \Big[ \frac{|x-x'|}{\theta} - \epsilon \Big]_+ \dif \pi^*(x,x') = 0  \Big\} \label{eq:lemma:Winfty2W1Aux2}\\
    & = \max\limits_{\rho \in \D, (s_i,s_j) \in \SetPair} \Big\{ \theta \colon  \sup\limits_{(x,x') \in \supp(\pi^*)}  \big( \frac{|x-x'|}{\theta} - \epsilon \big) = 0 \Big\} \nonumber\\
    & =  \frac{1}{\epsilon} \max_{\rho \in \D, (s_i,s_j) \in \SetPair} \sup_{(x,x')\in \supp(\pi^*)} |x-x'| . \nonumber
\end{align}
The infinum in \eqref{eq:lemma:Winfty2W1Aux1} is a Kantorovich optimal transport problem that determines $W_1$ distance.
Here, $\left[ \frac{|\cdot|}{\theta} - \epsilon \right]_+$ is a convex function for all $\theta$ and $\epsilon$.
Therefore, in \eqref{eq:lemma:Winfty2W1Aux2}, we substitute by the Kantorovich optimal transport plan $\pi^*$, i.e., the minimizer of $\inf_{\pi \in \Gamma(s_i,s_j)} \int \big[ \frac{|x-x'|}{\theta} - \epsilon \big]_+ \dif \pi(x,x')$. \qed

\section{EFFICIENT SOLUTION TO TWO PROBLEMS IN \cite{PufferfishWasserstein2017Song} BY LEMMA~\ref{lemma:Winfty2W1Laplace}}
\label{app:lemma:Winfty2W1LaplaceEx}

For the two examples in \citet[Section~3.1]{PufferfishWasserstein2017Song}, we show how to determine $\theta$ by Lemma~1.
Table~\ref{tab:ex1} shows the probabilities of $X$ conditioned on two instances $s_i$ and $s_j$ of $S$.
    \begin{table}[H]
    \caption{1st example of $P_{X|S}$} \label{tab:ex1}
        \begin{center}
        \begin{tabular}{ccccc}
            \hline\hline
            & $X = 1$ & $X=2$ & $X=3$ & $X=4$\\ \hline
            $P_{X|S}(\cdot|s_i,\rho)$ & 1/3  & 1/6 & 1/3 & 1/6 \\ \hline
            $P_{X|S}(\cdot|s_j,\rho)$   & 1/4  & 1/4 & 1/6 & 1/3 \\  \hline
            \end{tabular}
        \end{center}
    \end{table}
We first obtain the joint cumulative mass function (CMF) of the Kantorovich optimal transport plan $\pi^*((-\infty,x], (-\infty,x']) = \min\Set{F(x|s_i), F(x'|s_j)}$ in Table~\ref{tab:CDF1} and then the join probability mass function (PMF) $\pi^*(x, x')$ in Table~\ref{tab:PDF1}.\footnote{Recall that for $x_1,x_2, x'_1,x'_2 \in \X$ such that $x_1 < x_2$ and $x'_1 < x'_2$, $\pi^*([x_1,x_2],[x'_1,x'_2]) = \pi^*((-\infty,x_2],(-\infty,x'_2])- \pi^*((-\infty,x_1],(-\infty,x'_2]) - \pi^*((-\infty,x_2],(-\infty,x'_1]) + \pi^*((-\infty,x_1],(-\infty,x'_1])$. }
    \begin{table}[H]
    \caption{$\pi^*((-\infty,x],(-\infty,x'])$} \label{tab:CDF1}
        \begin{center}
        \begin{tabular}{ccccc}
            \hline\hline
            & $X' = 1$ & $X'=2$ & $X'=3$ & $X'=4$\\ \hline
            $X= 1$ & 1/4  & 1/3 & 1/3 & 1/3 \\ \hline
            $X= 2$ & 1/4  & 1/2 & 1/2 & 1/2 \\  \hline
            $X= 3$ & 1/4  & 1/2 & 2/3 & 5/6 \\  \hline
            $X= 4$ & 1/4  & 1/2 & 2/3 & 1 \\  \hline
            \end{tabular}
        \end{center}
    \end{table}

    \begin{table}[H]
    \caption{$\pi^*(x, x')$} \label{tab:PDF1}
        \begin{center}
        \begin{tabular}{ccccc}
            \hline\hline
            & $X' = 1$ & $X'=2$ & $X'=3$ & $X'=4$\\ \hline
            $X= 1$ & 1/4  & 1/12 & 0 & 0 \\ \hline
            $X= 2$   & 0  & 1/6 & 0 & 0 \\  \hline
            $X= 3$   & 0  & 0 & 1/6 & 1/6 \\  \hline
            $X= 4$   & 0  & 0 & 0 & 1/6 \\  \hline
            \end{tabular}
        \end{center}
    \end{table}
    We have $\sup_{(x,x')\in \supp(\pi^*)} |x-x'| = 1$. Applying Lemma~1, adding Laplace noise with $\theta = \frac{1}{\epsilon}$ attains pufferfish privacy.

    Table~\ref{tab:ex2} shows the second example of the conditional probabilities
    $P_{X|S}(\cdot|s_i,\rho)$ and $P_{X|S}(x|s_j,\rho)$, for which we have the joint CMF and PMF in Tables~\ref{tab:CDF2} and \ref{tab:PDF2}, respectively. By Lemma~2, adding Laplace noise with $\theta = \frac{2}{\epsilon}$ attains pufferfish privacy.
    \begin{table}[H]
    \caption{2nd example of $P_{X|S}$} \label{tab:ex2}
        \begin{center}
        \begin{tabular}{cccccc}
            \hline\hline
            & $X = 1$ & $X=2$ & $X=3$ & $X=4$ & $X=5$\\ \hline
            $P_{X|S}(\cdot|s_i,\rho)$ & 0.2  & 0.225 & 0.5 & 0.075 & 0 \\ \hline
            $P_{X|S}(\cdot|s_j,\rho)$ & 0  & 0.075 & 0.5 & 0.225 & 0.2 \\  \hline
            \end{tabular}
        \end{center}
    \end{table}

    \begin{table}[H]
    \caption{$\pi^*((-\infty,x],(-\infty,x'])$} \label{tab:CDF2}
        \begin{center}
        \begin{tabular}{cccccc}
            \hline\hline
            & $X' = 1$ & $X'=2$ & $X'=3$ & $X'=4$ & $X'=5$\\ \hline
            $X= 1$ & 0  & 0.075 & 0.2 & 0.2 & 0.2\\ \hline
            $X= 2$ & 0  & 0.075 & 0.425 & 0.425 & 0.425 \\  \hline
            $X= 3$ & 0  & 0.075 & 0.575 & 0.8 & 0.925 \\  \hline
            $X= 4$ & 0  & 0.075 & 0.575 & 0.8 & 1 \\  \hline
            $X= 5$ & 0  & 0.075 & 0.575 & 0.8 & 1 \\  \hline
            \end{tabular}
        \end{center}
    \end{table}

    \begin{table}[H]
    \caption{$\pi^*(x,x')$} \label{tab:PDF2}
        \begin{center}
        \begin{tabular}{cccccc}
            \hline\hline
            & $X' = 1$ & $X'=2$ & $X'=3$ & $X'=4$ & $X'=5$\\ \hline
            $X= 1$ & 0  & 0.075 & 0.125 & 0 & 0\\ \hline
            $X= 2$ & 0  & 0 & 0.225 & 0 & 0 \\  \hline
            $X= 3$ & 0  & 0 & 0.15 & 0.225 & 0.125 \\  \hline
            $X= 4$ & 0  & 0 & 0 & 0 & 0.075 \\  \hline
            $X= 5$ & 0  & 0 & 0 & 0 & 0 \\  \hline
            \end{tabular}
        \end{center}
    \end{table}

\section{PROOF of LEMMA~\ref{lemma:NindUser}}
\label{appB:lemma:NindUser}

In the $V$ independent user system, for the discriminative pair set $\SetPair_i = \Set{(S_i = a, S_i = b) \colon a, b \in \SA_i}$, the sensitivity for the query function $f$ and metric $d$ is
\[ \triangle_f(\SetPair_i) = \max_{a,b\in \SA_i} \max_{s_{-i}} d \big( f(s_i=a ,s_{-i}) - f(s_i=b,s_{-i}) \big) \]
for all $\rho \in \D$.
Using the probabilities derived in Section~\ref{appA:lemma:NindUser},
\begin{align*}
     \frac{P_{Y|S_i}(y|a,\rho)}{P_{Y|S_i}(y|b,\rho)}
     & = \frac{\int P_{N_{\theta}}(y- f(s_i=a, s_{-i})) P_{S_{-i}}(s_{-i}) \dif s_{-i}}{\int P_{N_{\theta}}(y- f(s_i=b, s_{-i})) P_{S_{-i}}(z_{-i}) \dif s_{-i}} \\
     & \leq \frac{\int P_{N_{\theta}}(y- f(s_i=b, s_{-i})) e^{\eta(\theta) d(f(s_i=a,s_{-i})-f(s_i=b,s_{-i})) } P_{S_{-i}}(s_{-i}) \dif s_{-i}}{\int P_{N_{\theta}}(y- f( s_i=b,s_{-i})) P_{S_{-i}}(s_{-i}) \dif s_{-i}}  \\
     &\leq e^{\eta(\theta) \triangle_f(\SetPair_i)}, \quad \forall a,b \in \SA_i.
\end{align*}
Therefore, $\theta = \eta^{-1}(\frac{\epsilon}{\triangle_f(\SetPair_i)})$ is a sufficient condition for $\frac{P_{Y|S_i}(y|a,\rho)}{P_{Y|S_i}(y|b,\rho)} \leq e^{\epsilon}$ for all $a,b \in \SA_i$, $y$ and $\rho$. This proves (a).

For $(\epsilon,\SetPair_i)$-pufferfish private $Y$,  we have $P_{Y|S_i}(y|a,\rho) \leq e^{\epsilon} P_{Y|S_i}(y|b,\rho)$ and $P_{Y|S_i}(y|b,\rho) \leq e^{\epsilon} P_{Y|S_i}(y|a)$ for all $a,b \in \SA_i$, $y$ and $\rho$.
Using the fact that $P_{Y|S_i}(y | \perp_i,\rho) = \int P_{Y|S_i}(y | s_i,\rho) P_{S_i}(s_i) \dif s_i$,
\begin{align*}
    \frac{P_{Y|S_i}(y|\perp_i,\rho)}{P_{Y|S_i}(y|a,\rho)} & = \frac{\int P_{Y|S_i}(y|s_i,\rho) P_{S_i}(s_i) \dif s_i}{P_{Y|S_i}(y|a,\rho)}  \\
    & \leq  \frac{\int e^{\epsilon} P_{Y|S_i}(y|a,\rho) P_{S_i}(s_i) \dif s_i}{P_{Y|S_i}(y|a,\rho)} = e^{\epsilon}, \\
    \frac{P_{Y|S_i}(y|a,\rho)}{P_{Y|S_i}(y|\perp_i,\rho)} & = \frac{P_{Y|S_i}(y|a,\rho)}{\int P_{Y|S_i}(y|s_i,\rho) P_{S_i}(s_i) \dif s_i}  \\
    & = \frac{\int P_{Y|S_i}(y|a,\rho) P_{S_i}(s_i) \dif s_i}{\int P_{Y|S_i}(y|s_i,\rho) P_{S_i}(s_i) \dif s_i} \\
    & \leq \frac{\int e^{\epsilon} P_{Y|S_i}(y|s_i,\rho) P_{S_i}(s_i) \dif s_i}{\int P_{Y|S_i}(y|s_i,\rho) P_{S_i}(s_i) \dif s_i} = e^{\epsilon},
\end{align*}
for all $a \in \SA_i$, $y$ and $\rho$, i.e., $Y$ is $(\epsilon,\SetPair_{\perp_i})$-pufferfish private. This proves (b). \qed

\section{PROOF of THEOREM~\ref{theo:W1Exponential}}
\label{app:theo:W1Exponential}

We derive the following results for each $\rho \in \D$.
For any pair $(s_i,s_j) \in \SetPair$,
\begin{align} \label{eq:W1Exponential:SuffCond1}
     P_{Y|S}(y|s_i,\rho) - e^\epsilon P_{Y|S}(y|s_j,\rho) & = \int P_{N_{\theta}}(y-x) P_{X|S}(x|s_i,\rho) \dif x - e^{\epsilon} \int P_{N_{\theta}}(y-x') P_{X|S}(x'|s_j,\rho) \dif x' \nonumber \\
     & =  \int \big( P_{N_{\theta}}(y-x) - e^{\epsilon} P_{N_{\theta}}(y-x') \big) \dif \pi(x,x')  \nonumber\\
     & \leq \int P_{N_{\theta}}(y-x') \big( e^{\eta(\theta) d(x-x')}  - e^{\epsilon} \big)  \dif \pi(x,x'), \qquad \forall y.
\end{align}
Requesting $e^{\eta(\theta) d(x-x')}  - e^{\epsilon} \leq 0 $ for each pair of $x$ and $x'$, we derive a sufficient condition
\begin{equation} \label{eq:W1Exponential:SuffCond}
    \inf_{\pi \in \Gamma(s_i,s_j)} \sup_{(x,x') \in \supp(\pi)} \big( \eta(\theta) d(x-x') - \epsilon \big) \leq 0
\end{equation}
for $\frac{P_{Y|S}(y|s_i,\rho)}{P_{Y|S}(y|s_j,\rho)} \leq e^{\epsilon}$. Note that the infimum in \eqref{eq:W1Exponential:SuffCond} is for the purpose of searching the minimum value of $\theta$ (over all couplings) that holds the sufficient condition, knowing $\VAR[N_{\theta}] \propto \theta$.

The sufficient condition \eqref{eq:W1Exponential:SuffCond} is equivalent to
\begin{equation} \label{eq:W1Exponential:SuffCond2}
   \inf_{\pi \in \Gamma(s_i,s_j)}  \int \big[ \eta(\theta) d(x-x') - \epsilon \big]_+ \dif \pi(x,x')
    = \int  \big[ \eta(\theta) d(x-x') - \epsilon \big]_+ \dif \pi^*(x,x') \leq 0
\end{equation}
where $\pi^*$ denotes the Kantorovich optimal transport plan. We further convert \eqref{eq:W1Exponential:SuffCond2} to
\begin{align}
    \int  \big[ \eta(\theta) d(x-x') - \epsilon \big]_+ \dif \pi^*(x,x') \leq 0
    & \Longrightarrow \sup_{(x,x') \in \supp(\pi^*)} \eta(\theta) d(x-x') - \epsilon  \leq 0 \nonumber\\
    & \Longrightarrow \eta(\theta) \leq \epsilon / \sup_{(x,x') \in \supp(\pi^*)} d(x-x'). \label{eq:W1Exponential:SuffCond3}
\end{align}
For invertible $\eta$ that is nonincreasing in $\theta$, the minimum value of $\theta$ that holds the inequality \eqref{eq:W1Exponential:SuffCond3} is
\[ \eta^{-1} \Big( \epsilon / \sup_{(x,x') \in \supp(\pi^*)} d(x-x') \Big).\]
Taking the maximum of this value over all $(s_i,s_j) \in \SetPair$ and $\rho \in \D$, we have the sufficient condition
\begin{equation} \label{eq:W1Exponential:SuffCond4}
    \theta = \max_{\rho \in \D, (s_i,s_j) \in \SetPair} \eta^{-1}  \Big(  \epsilon / \sup_{(x,x') \in \supp(\pi^*)} d(x-x')  \Big).
\end{equation}
To make $\frac{P_{Y|S}(y|s_j,\rho)}{P_{Y|S}(y|s_i,\rho)} \leq e^{\epsilon}$, we have the sufficient condition
$ \inf_{\pi \in \Gamma(s_j,s_i)} \sup_{(x',x) \in \supp(\pi)} \big( \eta(\theta) d(x'-x) - \epsilon \big) \leq 0 $
the minimizer for which is $\pi^{*\intercal}$ such that $\pi^{*\intercal}(x',x)=\pi^{*}(x,x')$.
Based on the symmetry property of $d$, i.e., $d(x-x') = d(x'-x)$, $\sup_{(x,x') \in \pi^*} d(x-x') = \sup_{(x',x) \in \pi^{*\intercal}} d(x'-x)$, i.e., \eqref{eq:W1Exponential:SuffCond3} is also a sufficient condition for $\frac{P_{Y|S}(y|s_j,\rho)}{P_{Y|S}(y|s_i,\rho)} \leq e^{\epsilon}$. This completes the proof. \qed

\section{PROOF of LEMMA~\ref{lemma:NindUser} BY THEOREM~\ref{theo:W1Exponential}}
\label{lemmaInt:NindUser}

For the $V$ independent user system, the result that $(\epsilon, \SetPair)$-pufferfish privacy can be attained by calibrating noise to the sensitivity of the query function $f$ in Lemma~\ref{lemma:NindUser} is explained by the support of the Kantorovich optimal transport plan $\pi^*$ in Theorem~\ref{theo:W1Exponential}.

To obtain the cumulative density distribution (CDF) $\pi^*((-\infty,x], (-\infty,x']) = \min\Set{F_{X|S_i}(x|a,\rho), F_{X|S_i}(x'|\perp_i, \
\rho)}$, first consider any $x$ and $x'$ such that $x > x'$. Let $\delta = d(x-x')$. Then, $x = x' + d^{-1}(\delta)$.
For any $a \in \SA_i$, we derive the CDFs $F_{X|S_i}(x|a, \rho)$ and $F_{X|S_i}(x|\perp_i, \rho)$ for the conditional probabilities $P_{X|S_i}(x|a, \rho)$ and $P_{X|S_i}(x|\perp_i, \rho)$, respectively, as follows.

\begin{align*}
    F_{X|S_i}(x|a, \rho) & = \int_{-\infty}^{x} P_{X|S_i}(l|a,\rho) \dif l  \\
    & = \int_{-\infty}^{x} \Pr(f(S_i=a,S_{-i}) = l) \dif l  \\
    & = \int_{-\infty}^{x' + d^{-1}(\delta)} \Pr(f(S_i=a,S_{-i}) = l) \dif l \\
    & = \int_{-\infty}^{x'} \Pr(f(S_i=a,S_{-i}) = l + d^{-1}(\delta)) \dif l, \\
    F_{X|S_i}(x|\perp_i,\rho) & = \int_{-\infty}^{x} P_{X|S_i}(l | \perp_i,\rho) \dif l \\
    & = \int_{-\infty}^{x}  \int P_{X|S_i}(l | s_i,\rho) P_{S_i}(s_i) \dif s_i \dif l \\
    %& = \int \Big( \int_{-\infty}^{x}  P_{X|S_i}(l | s_i,\rho) \dif l \Big) P_{S_i}(s_i) \dif s_i\\
    & = \int \Big( \int_{-\infty}^{x}  \Pr(f(S_i=s_i,S_{-i}) = l) \dif l \Big) P_{S_i}(s_i) \dif s_i.
\end{align*}
Comparing $F_{X|S_i}(x|a,\rho)$ and $F_{X|S_i}(x'|b,\rho)$, we have
\begin{align*}
    F_{X|S_i}(x|a,\rho) - F_{X|S_i}(x'|b,\rho)
    & = \int_{-\infty}^{x'} \Pr(f(S_i=a,S_{-i}) = l + d^{-1}(\delta)) \dif l  - \int_{-\infty}^{x'} \Pr(f(S_i=b,S_{-i}) = l) \dif l\\
    & = \int_{-\infty}^{x'} \Big( \Pr(f(S_i=a,S_{-i}) = l + d^{-1}(\delta)) - \Pr(f(S_i=b,S_{-i}) = l) \Big) \dif l \\
    & < 0 , \qquad \forall \delta > \triangle_f(\SetPair_i).
\end{align*}
Similarly, for all $x$ and $x'$ such that $x < x'$, due to the symmetry property of $d$, $d(x-x') = d(x'-x) = \delta \Longrightarrow x' = x + d^{-1}(\delta)$ and
\begin{multline*}
    F_{X|S_i}(x|a,\rho) - F_{X|S_i}(x'|b,\rho) \\
    = \int_{-\infty}^{x'} \Big( \Pr(f(S_i=a,S_{-i}) = l ) - \Pr(f(S_i=b,S_{-i}) = l + d^{-1}(\delta)) \Big) \dif l > 0, \qquad \forall \delta > \triangle_f(\SetPair_i).
\end{multline*}
That is, for all $x,x'$ such that $d(x-x') > \triangle_f(\SetPair_i)$, $\pi^*((-\infty,x], (-\infty,x']) = \min\Set{F_{X|S_i}(x|a,\rho), F_{X|S_i}(x'|\perp_i, \rho)}$ is independent on either $x$ or $x'$ and so
\[ \pi^*(x,x') = \frac{\dif^2}{\dif x  \dif x'} \pi^*((-\infty,x], (-\infty,x']) = 0, \qquad \forall x,x' \colon  d(x-x') > \triangle_f(\SetPair_i). \]
Equivalently,
\[\supp(\pi^*) \subseteq \big\{(x,x') \colon d(x-x') \leq \triangle_f(\SetPair_i) \big\} \]
and so
\[ \sup_{(x,x') \in \supp(\pi^*)} d(x-x') \leq  \triangle_f(\SetPair_i).  \]
Therefore, by Theorem~\ref{theo:W1Exponential}, adding noise $\theta$ with
\[ \theta = \eta^{-1}  \Big(  \frac{\epsilon}{\triangle_f(\SetPair_i)}  \Big). \]
attains $(\epsilon,\SetPair_i)$-pufferfish privacy, which proves Lemma~2(a).

Consider the discriminative pair set $\SetPair_{\perp_i} = \Set{(S_i = a, S_i = \perp_i) \colon a \in \SA_i}$.
For all $a \in \SA_i$, we have
\begin{align*}
    F_{X|S_i}(x|a,\rho) - F_{X|S_i}(x'|\perp_i,\rho)
    & = \int \Big( \int_{-\infty}^{x}  \Pr(f(S_i=a,S_{-i} = l)) \dif l - \int_{-\infty}^{x'} \Pr(f(S_i=s_i,S_{-i} = l)) \dif l \Big) P_{S_i}(s_i) \dif s_i\\
    & = \int \Big( F_{X|S_i}(x|s_i,\rho) - F_{X|S_i}(x|a,\rho) \Big) P_{S_i}(s_i) \dif s_i.
\end{align*}
Again, we have $\supp(\pi^*) \subseteq \big\{(x,x') \colon d(x-x') \leq \triangle_f(\SetPair_i) \big\} $ and $\theta = \max_{(s_i,s_j) \in \SetPair} \eta^{-1}  \big(  \epsilon / \triangle_f(\SetPair_i)  \big)$. This proves Lemma~2(b).

\subsection{Separable Query Function}
In addition, we obtain an extra result for the separable query function.
Assume $f$ is separable, i.e., $f(s) = \sum_{i \in \N} f_i(s_i), \forall s = (s_i \colon i \in \N)$. One example is the counting query $f(s) = \sum_{i \in \N} s_i$.

For any $a,b\in \SA_i$,
\[ f(s_i=a,s_{-i}) - f(s_i=b,s_{-i}) = f_i(a) - f_i(b), \qquad \forall s_{-i}. \]
Let $\delta = f_i(a) - f_i(b) $, we have
\begin{align*}
    F_{X|S_i}(x|a,\rho)
    &= \int_{-\infty}^{x} \Pr( f(S_i = a, S_{-i}) = l  ) \dif l \\
    & = \int_{-\infty}^{x} \Pr( f(S_i = b, S_{-i}) = l - \delta  ) \dif l \\
    & = \int_{-\infty}^{x-\delta} \Pr( f(S_i = b, S_{-i}) = l) \dif l
\end{align*}
and so
\begin{align*}
    F_{X|S_i}(x|a,\rho) - F_{X|S_i}(x'|b,\rho)
    & = \int_{-\infty}^{x-\delta} \Pr( f(S_i = b, S_{-i}) = l) \dif l  - \int_{-\infty}^{x'} \Pr( f(S_i = b, S_{-i}) = l) \dif l\\
    & \begin{cases}
          > 0 & x-x' > \delta \\
          < 0 & x-x' < \delta \\
          =0 & x-x' = \delta
      \end{cases}.
\end{align*}
Therefore,
\begin{align*}
    \pi^*(x,x')
    &= \frac{\dif^2}{\dif x  \dif x'} \pi^*(x,x')\\
    &= \frac{\dif^2}{\dif x  \dif x'} \min\Set{F_{X|S_i}(x|a,\rho), F_{X|S_i}(x'|b,\rho)} = 0 , \qquad \forall x,x' \colon x-x' \neq \delta.
\end{align*}
That is, the support of $\pi^*$ is
\[ \supp(\pi^*) = \big\{(x,x') \colon x-x' = f_i(a) - f_i(b) \big\}.  \]

Also note that for separable query functions $f$,
\begin{align*}
    \triangle_f(\SetPair_i) &= \max_{a,b\in \SA_i} \max_{s_{-i}} d \big( f(s_i=a ,s_{-i}) - f(s_i=b,s_{-i}) \big) \\
    &= \max_{a,b\in \SA_i} d ( f_i(a) - f_i(b)).
\end{align*}
is independent of $S_{-i}$ and
\[ \theta = \eta^{-1}  \Big(  \frac{\epsilon}{\max_{a,b\in \SA_i} d ( f_i(a) - f_i(b))}  \Big). \]

\section{PROOF of THEOREM~\ref{theo:SuffCondRelax}}
\label{app:theo:SuffCondRelax}

The proof starts with the proposition below.
\begin{proposition} \label{prop:SuffCondRelax}
    If  there exists a nonnegative function $D_{\epsilon}(\cdot; \theta)$ such that
    \begin{equation}  \label{eq:SuffCond}
		e^{\eta(\theta) d(z)}- e^{\epsilon} \leq D_{\epsilon}(z; \theta) , \quad \forall \theta,z
	\end{equation}
    and $D_{\epsilon}(\cdot; \theta)$ is convex in $z$ and nonincreasing in $\theta$,
    $(\epsilon,\SetPair)$-pufferfish privacy is attained by adding noise $N_{\theta}$ with any $\theta$ that holds the inequalities
    \begin{subequations} \label{eq:SuffCondRelax}
        \begin{align}
            & \int e^{\eta(\theta) d(x-x')} \pi^*(x,x') \dif x \leq e^{\epsilon} p(x'|s_j), \quad \forall x' \label{eq:SuffCondRelax1} \\
            & \int e^{\eta(\theta) d(x-x')} \pi^*(x,x') \dif x \leq e^{\epsilon} p(x|s_i), \quad \forall x \label{eq:SuffCondRelax2}
        \end{align}
    \end{subequations}
 for all $(s_i,s_j) \in \SetPair$.
\end{proposition}
\begin{proof}
    Recall the inequality $P_{Y|S}(y|s_i,\rho) - e^\epsilon P_{Y|S}(y|s_j,\rho) \leq \int P_{N_{\theta}}(y-x') \big( e^{\eta(\theta) d(x-x')}  - e^{\epsilon} \big)  \dif \pi(x,x'), \forall y$ obtained in \eqref{eq:W1Exponential:SuffCond1}.
    Upon the condition \eqref{eq:SuffCond}, we have
	\begin{align}
		P_{Y|S}(y|s_i,\rho) - e^\epsilon P_{Y|S}(y|s_j,\rho) &\leq \int P_{N_{\theta}}(y-x') \big( e^{\eta(\theta) d(x-x')}  - e^{\epsilon} \big)  \dif \pi(x,x') \label{eq:prop:SuffCondRelaxAux1} \\
        &\leq \int P_{N_{\theta}}(y-x')  D_{\epsilon}(x-x'; \theta)   \dif \pi(x,x')  \label{eq:prop:SuffCondRelaxAux3}\\
		& \leq \int D_{\epsilon}(x-x'; \theta)   \dif \pi(x,x'). \label{eq:prop:SuffCondRelaxAux4}
	\end{align}
    The inequalities above holds for all $\pi \in \Gamma(s_i,s_j)$. That is, for any coupling $\pi \in \Gamma(s_i,s_j)$, if $\int D_{\epsilon}(x-x'; \theta)   \dif \pi(x,x') \leq 0$, then $\frac{P_{Y|S}(y|s_i)}{P_{Y|S}(y|s_j)} \leq e^{\epsilon}$.

    For $D_{\epsilon}$ nonincreasing in $\theta$, we take the infinum of the integral in \eqref{eq:prop:SuffCondRelaxAux4} over all couplings and request
    \begin{equation} \label{eq:prop:SuffCondRelaxAux5}
        \inf_{\pi \in \Gamma(s_i, s_j)} \int D_{\epsilon} (x-x';\theta) \dif \pi (x,x') \leq 0.
    \end{equation}
    The purpose is to find the smallest value of $\theta$ that holds the sufficient condition.
    It is clear that for $D_{\epsilon}(z;\theta)$ being convex in $z$ for all $\theta$, the minimizer of \eqref{eq:prop:SuffCondRelaxAux5} is the Kantorovich optimal transport plan $\pi^*$, i.e., \eqref{eq:prop:SuffCondRelaxAux5} is equivalent to $\int D_{\epsilon} (x-x';\theta) \dif \pi^* (x,x') \leq 0 $. Instead, we request the integral in \eqref{eq:prop:SuffCondRelaxAux1} under $\pi^*$ to be nonpositive,\footnote{This is also for the purpose of searching the minimum value of $\theta$ that holds the sufficient condition.} i.e.,
    \begin{align*}
		\int P_{N_{\theta}}(y-x') \big( e^{\eta(\theta) d(x-x')}  - e^{\epsilon} \big)  \dif \pi^*(x,x') &= \iint P_{N_{\theta}}(y-x) \big( e^{\eta(\theta) d(x-x')} - e^{\epsilon} \big) \pi^*(x,x') \dif x \dif x' \\
        &= \int P_{N_{\theta}}(y-x')  \int \big( e^{\eta(\theta) d(x-x')} - e^{\epsilon}\big) \pi^*(x,x')  \dif x \dif x' \leq 0.
	\end{align*}
    A sufficient condition to hold this inequality is to make $\int \big( e^{\eta(\theta) d(x-x')} - e^{\epsilon}\big) \pi^*(x,x')  \dif x \leq 0$ for all $x'$, which is equivalent to \eqref{eq:SuffCondRelax1}.

    Similarly, a sufficient condition for $\frac{P_{Y|S}(y|s_j)}{P_{Y|S}(y|s_i)} \leq e^{\epsilon}$ is
    \begin{equation} \label{eq:prop:SuffCondRelaxAux6}
        \int e^{\eta(\theta) d(x'-x)} \pi^{*\intercal}(x',x) \dif x \leq e^{\epsilon} p(x|s_i), \quad \forall x ,
    \end{equation}
    where $\pi^{*\intercal}(x',x) = \pi^*(x,x'), \forall x,x'$. Due to the symmetry property of $d$, the condition \eqref{eq:prop:SuffCondRelaxAux6} is equivalent to \eqref{eq:SuffCondRelax2}.
\end{proof}

Note that the value of $\theta$ in Proposition~\ref{prop:SuffCondRelax} is not determined by the upper bound function $D_{\epsilon}(z; \theta)$ in \eqref{eq:SuffCond}.
That is, we only require the existence of such a function $D_{\epsilon}(z; \theta)$ regardless of the tightness of this upper bound.
We show below that one example of this upper bound function is the piecewise linear function $[\cdot]_+ = \max\Set{\cdot,0}$.

For the exponential mechanism, we have
\begin{align*}
    e^{\eta(\theta) d(z)} - e^{\epsilon} & \leq \frac{e^{\eta(\theta) d(z)} + e^{\epsilon}}{2} \big[ \eta(\theta) d(z) - \epsilon \big]_{+} \\
                                         & \leq \frac{e^{ M \eta(\theta) } + e^{\epsilon}}{2} \big[ \eta(\theta) d(z) - \epsilon \big]_{+},
\end{align*}
where $M = \max_{x,x'} d(x-x')$ assuming $d(\cdot)$ is a bounded measure.\footnote{Here, we apply the inequality for exponential function: $e^{\frac{x+y}{2}} \leq \frac{e^{y} - e^{x}}{y-x} \leq \frac{e^x + e^y}{2}$ for all $x,y \in \Real$.}
Here, $[f(\cdot)]_{+}$ for convex $f$ is convex and $\big[ \eta(\theta) d(z) - \epsilon \big]_{+}$ is nonincreasing in $\theta$ since $\eta \propto \frac{1}{\theta}$.

Consider \eqref{eq:SuffCondRelax1}. For $\eta(\theta)$ nonincreasing in $\theta$, $\inf \Set{\theta \colon \int e^{\eta(\theta) d(x-x')} \pi^*(x,x') \dif x \leq e^{\epsilon} p(x'|s_j) }$ equals to the value of $\theta$ that holds the equality
$$ \int e^{\eta(\theta) d(x-x')} \pi^*(x,x') \dif x = e^{\epsilon} p(x'|s_j) $$
for each $x'$. Taking also the equality of \eqref{eq:SuffCondRelax2}, we have Theorem~2. \qed

\section{PROOF of THEOREM~\ref{theo:GaussN_AISTATS}}
\label{app:theo:GaussN_AISTATS}

For zero-mean Gaussian noise, $P_{N_\theta}(z) = \frac{1}{\sqrt{2\pi} \theta} e^{-\frac{z^2}{2\theta^2}}$ and
\begin{align} \label{eq:GaussAux1}
     P_{Y|S}(y|s_i,\rho) - e^\epsilon P_{Y|S}(y|s_j,\rho) & = \int P_{N_{\theta}}(y-x) P_{X|S}(x|s_i,\rho) \dif x - e^{\epsilon} \int P_{N_{\theta}}(y-x') P_{X|S}(x'|s_j,\rho) \dif x' \nonumber \\
     & =  \int \big( P_{N_{\theta}}(y-x) - e^{\epsilon} P_{N_{\theta}}(y-x') \big) \dif \pi(x,x')  \nonumber\\
     & = \int \frac{1}{\sqrt{2\pi} \theta} \big( e^{-\frac{(y-x)^2}{2\theta^2}}  - e^{\epsilon-\frac{(y-x')^2}{2\theta^2}}\big)  \dif \pi(x,x'), \qquad \forall y.
\end{align}
This equality holds for the Kantorovich optimal transport plan $\pi^*$.\footnote{Choosing $\pi^*$ will necessarily reduce the sensitivity $\triangle$. }
In this case, to have \eqref{eq:GaussAux1}$\leq 0$, we only need to request
\begin{align} \label{eq:GaussAux2}
    \frac{(y-x')^2-(y-x)^2}{2\theta^2} & =  \frac{(y-x + x -x')^2-(y-x)^2}{2\theta^2} \nonumber\\
                        & = \frac{2(x -x')(y-x)+(x-x')^2}{2\theta^2}  \nonumber \\
                        & \leq \frac{2\triangle|y-x| + \triangle^2}{2\theta^2} \leq \epsilon,
\end{align}
where $\triangle = \sup_{(x,x') \in \supp(\pi^*)} |x-x'|$.
We follow the same approach in \citet[Appendix~A]{Dwork2014book} to prove (a). Let $\theta = \frac{\triangle}{\epsilon} c$, where $c \geq 0$. Then, $\frac{\theta}{\triangle} = \frac{c}{\epsilon}$.
Rewriting inequality \eqref{eq:GaussAux2} to
\begin{align*}
    \frac{\triangle|y-x|}{\theta^2} \leq \epsilon -\frac{\triangle^2}{2\theta^2}  \quad
    &\Longrightarrow \quad  \frac{|y-x|}{\theta} \leq \epsilon \frac{\theta}{\triangle} - \frac{\triangle}{2\theta} \\
    &\Longrightarrow \quad  \frac{|y-x|}{\theta} \leq c - \frac{\epsilon}{2c}.
\end{align*}
Here, $\frac{Y-X}{\theta}$ follows standard normal distribution.

Recall that for standard normal distributed random variable $Z$, we have a lower bound on tail probability: $\Pr(Z > t) > \frac{1}{\sqrt{2 \pi } t} e^{-\frac{t^2}{2}}$.
For $t \geq 0 $, we are seeking the condition on $t$ that holds inequality $\Pr(|Z| > t) < \delta$, which (due to the symmetry of Gaussian distribution) can be enforced on the positive range:
\begin{align} \label{eq:GaussAux3}
    \Pr(Z > t) < \frac{1}{\sqrt{2 \pi } t} e^{-\frac{t^2}{2}} < \frac{\delta}{2}
    \quad &\Longrightarrow \quad t e^{\frac{t^2}{2}} > \sqrt{\frac{2}{\pi}} \frac{1}{\delta}  \nonumber \\
     &\Longrightarrow \quad \log t + \frac{t^2}{2} > \log \sqrt{\frac{2}{\pi}} \frac{1}{\delta}
\end{align}
So, for $Z = \frac{Y-X}{\theta}$, $t = c - \frac{\epsilon}{2c}$ with $c \geq \sqrt{\frac{\epsilon}{2}}$, we need to determine $c$ such that
\begin{align}
    \underbrace{\log \big( c - \frac{\epsilon}{2c} \big)}_{A} + \frac{1}{2} \underbrace{ \big( c^2 - \epsilon + \frac{\epsilon^2}{4c}}_{B} \big) > \log \sqrt{\frac{2}{\pi}} \frac{1}{\delta}
\end{align}
For $\epsilon \leq 1$, we set $c \geq \frac{3}{2}$ to have $A>0$, for which, $B \geq c^2 - \frac{8}{9}$. We instead request $c^2 - \frac{8}{9} \geq 2 \log \sqrt{\frac{2}{\pi}} \frac{1}{\delta}$ and have $c^2 \geq 2 \log (1.25 / \delta)$.

We use the inequality $\log t \leq t-1, \forall t > 0$ to prove $(b)$. Alternative to \eqref{eq:GaussAux3}, request
\begin{align}
    \log t + \frac{t^2}{2} \geq \log t + \log \frac{t^2}{2} + 1 > \log \sqrt{\frac{2}{\pi}} \frac{1}{\delta}
    \quad &\Longrightarrow \quad \log \frac{t^3}{2} \sqrt{\frac{\pi}{2}} \delta > -1 \\
    &\Longrightarrow \quad  t > \left( \frac{2}{e} \right)^{\frac{1}{3}} \left( \frac{2}{\pi}\right)^{\frac{1}{6}} \delta^{-\frac{1}{3}}.
\end{align}
As $ (2/e)^{\frac{1}{3}} (2/\pi)^{1/6} = 0.8373$, we need to have $t > 0.84 \delta^{-\frac{1}{3}} $.
For $Z = \frac{Y-X}{\theta}$ and $t = c - \frac{\epsilon}{2c}$ with $c \geq \sqrt{\frac{\epsilon}{2}}$,
\begin{align}
    c - \frac{\epsilon}{2c} > 0.84 \delta^{-\frac{1}{3}}
    \quad &\Longrightarrow \quad c^2 - 0.84 \delta^{-\frac{1}{3}} c - \frac{\epsilon}{2} > 0 \\
    &\Longrightarrow \quad  c > 0.41 \delta^{-\frac{1}{3}} + \sqrt{ (0.41\delta^{-\frac{1}{3}})^2 + \frac{\epsilon}{2} }
\end{align}

This proves (b). \qed

\end{document}